\newif\ifnotes
\notesfalse

\documentclass{article}
\usepackage{fullpage}
\usepackage{amsmath}
\usepackage{mathtools}
\usepackage{fullpage}
\usepackage{amsfonts}
\usepackage{amssymb}
\usepackage{amsthm}
\usepackage{bm}
\usepackage{url}
\usepackage{color}
\usepackage{xspace}
\usepackage{xcolor}
\usepackage{tikz-cd}
\usepackage{thmtools}
\usepackage{thm-restate}
\usepackage{extarrows}

\usepackage{soul}
\usepackage{mdframed}

\usepackage[margin=4mm,small,labelfont=bf]{caption}
\usepackage{algorithm,algpseudocode}
\usepackage{indentfirst}
\usepackage{twoopt}

\definecolor{carnelian}{rgb}{0.7, 0.11, 0.11}
\definecolor{darkcerulean}{rgb}{0.03, 0.27, 0.49}

\usepackage[colorlinks=true,citecolor=darkcerulean,linkcolor=carnelian]{hyperref} % for ECCC hyperlinks add option dvips

\usepackage[capitalize]{cleveref}

\allowdisplaybreaks

\newtheorem{theorem}{Theorem}[section]
\newtheorem{lemma}[theorem]{Lemma}
\newtheorem{definition}[theorem]{Definition}

\newtheorem{claim}[theorem]{Claim}
\newtheorem{proposition}[theorem]{Proposition}

\theoremstyle{remark}
\newtheorem{remark}[theorem]{Remark}

\newcommand{\F}{\mathbb{F}}
\renewcommand{\P}{\mathbb{P}}
\newcommand{\Prj}{\P^1}

\newcommand{\nice}[1]{{\sf{#1}}\xspace}

\newcommand{\poly}{\mathsf{poly}}

\newcommand{\zr}{^{(0)}}

\newcommand{\one}{^{(1)}}
\newcommand{\two}{^{(2)}}
\newcommand{\ii}{^{(i)}}
\newcommand{\jj}{^{(j)}}

\newcommand{\iip}{^{(i+1)}}

\newcommand{\rounds}{\nice{k}}
\newcommand{\fin}{^{(\rounds)}}
\newcommand{\ALG}{\mathsf{ALG}}

\newcommand{\DIV}{\mathsf{DIV}}

\newcommand{\DEGREE}{\mathsf{DEGREE}}

\newcommand{\EXTEND}{\mathsf{EXTEND}}
\newcommand{\MEXTEND}{\mathsf{MEXTEND}}
\newcommand{\MULT}{\mathsf{MULT}}
\newcommand{\ENTER}{\mathsf{ENTER}}
\newcommand{\REDC}{\mathsf{REDC}}
\newcommand{\EXIT}{\mathsf{EXIT}}
\newcommand{\MODV}{\mathsf{MOD}}
\newcommand{\CRT}{\mathsf{CRT}}
\newcommand{\rem}{\mathsf{\:rem\:}}
\newcommand{\FFTree}{\ensuremath{\mathsf{FFTree}}\xspace}
\newcommand{\FFTrees}{\ensuremath{\mathsf{FFTrees}}\xspace}
\newcommand{\FFForest}{\ensuremath{\mathsf{FFForest}}\xspace}
\newcommand{\nn}{\nice{K}}
\newcommand{\nnq}{\widehat{\nn}_q}
\newcommand{\calF}{\mathcal{F}}
\newcommand{\Z}{{\mathbb Z}}
\newcommand{\sym}{\nice{Sym}}
\newcommand{\mmod}{\!\!\mod}

\DeclarePairedDelimiter{\ceil}{\lceil}{\rceil}

\DeclarePairedDelimiter{\parens}()

\newcommand{\fq}{\mathbb{F}_q}
\newcommand{\fqbar}{\overline\fq}

\newcommand{\EVT}[2]{\left\langle{ #1 \wr #2 }\right\rangle}

\newcommand{\canon}[2]{\left\langle #1 \right\rangle_{#2}}

\newcommand{\eli}[1]{\ifnotes{\color{cyan} Eli: #1}\xspace\fi}
\newcommand{\swastik}[1]{\ifnotes{\color{red} Swastik: #1}\xspace\fi}

\begin{document}

\title{Elliptic Curve Fast Fourier Transform (ECFFT) Part I: \\ Fast Polynomial Algorithms over all Finite Fields}
\author{Eli Ben-Sasson\thanks{StarkWare Industries Ltd. {\tt \{eli,dancar,david\}@starkware.co}}
\and
Dan Carmon\footnotemark[1]
\and Swastik Kopparty \thanks{Department of Mathematics and Department of Computer Science, University of Toronto. Research supported in part by NSF grants CCF-1540634 and CCF-1814409, at Rutgers University. {\tt swastik.kopparty@gmail.com}}
\and David Levit\footnotemark[1]}

\maketitle

\begin{abstract}
Over finite fields $\fq$ containing a root of unity of smooth order $n$ (smoothness means  $n$ is the product
of small primes),
the Fast Fourier Transform (FFT) leads to the fastest known algebraic algorithms for many basic polynomial
operations, such as multiplication, division, interpolation and multi-point evaluation.
These operations can be computed by constant fan-in arithmetic circuits over $\fq$ of quasi-linear size;
specifically, $O(n\log n)$ for multiplication and division, and $O(n \log^2 n)$ for interpolation and evaluation.
%$O(n\log n)$ for the first two listed operations, and $O(n \log^2 n)$ for the next two.
%\dan{Suggest changing the part with ``first two, next two'', which I find difficult to parse, to:
%specifically, $O(n\log n)$ for multiplication and division, and $O(n \log^2 n)$ for interpolation and evaluation.}
%\eli{accepted}

However, the same operations over fields with no smooth order root of unity suffer from an asymptotic slowdown, typically due to the need to introduce ``synthetic'' roots of unity to enable the FFT.
The classical algorithm of Sch\"{o}nhage and Strassen~\cite{ScSt71} incurred a multiplicative slowdown factor of $\log\log n$ on top of the smooth case. Recent
remarkable results of Harvey, van der Hoeven and Lecerf~\cite{HHL17,HH19} dramatically reduced this multiplicative overhead to $\exp(\log^*(n))$.

We introduce a new approach to fast algorithms for polynomial operations over all large finite fields.
The key idea is to replace the group of roots of unity with a
set of points $L\subset \fq$
suitably related to a well-chosen elliptic curve group over $\fq$ (the set $L$ itself is \emph{not} a group).
The key advantage of this approach is that elliptic curve groups can be of {\em any} size in the Hasse--Weil interval $[q + \pm 2\sqrt{q} + 1]$
and thus can have subgroups of large, smooth order, which an FFT-like divide and conquer algorithm can exploit.
%\david{sub-}groups (and their projection sets \david{the projection of the group itself is not exactly half the size so it does not divide it, that's why we need a coset})
%can be of \emph{any} size $n$ that divides an integer in the Hasse--Weil interval $[q+1 \pm 2\sqrt{q}]$;
Compare this with multiplicative subgroups over $\fq$ whose order must divide $q-1$.
By analogy, our method extends the standard, multiplicative FFT in a similar way to how
Lenstra's elliptic curve method~\cite{Lenstra} extended Pollard's $p-1$ algorithm~\cite{p-1}
for factoring integers.

%\dan{Maybe add here the analogy, (this paper):FFT = ECM:($p-1$ method for factoring).}
%\eli{If we want to add this, should appear at end of last paragraph.} \eli{On more thought, I don't see the analogy to be simple enough to warrant adding to the already long abstract.}
%\david{The analogy here fits like a glove, it is actually the original inspiration}
%\dan{Wrote a short sentence here, also an extended paragraph in 1.3}
%\eli{terrific phrasing, closing this}

For polynomials represented by their evaluation over subsets of $L$, we show that multiplication, division, degree-computation, interpolation, evaluation and Reed--Solomon encoding (also known as low-degree extension) {\em with fixed evaluation points} can all be computed with arithmetic circuits of size similar to what is achievable with the classical FFTs when the field size $q$ is special.
For several problems, this yields the asymptotically smallest known arithmetic circuits even in the standard monomial representation of polynomials.

The efficiency of the classical FFT follows from using the $2$-to-$1$ squaring map to reduce the evaluation set of roots of unity of order $2^k$ to similar groups of size $2^{k-i}, i>0$. Our algorithms operate similarly, using isogenies of elliptic curves with kernel size $2$ as $2$-to-$1$ maps to reduce $L$ of size $2^k$ to sets of size $2^{k-i}$ that are, like $L$, suitably related to
 elliptic curves, albeit different ones.

%that is a  projection of points in \david{a coset of}

%\david{sub-}groups (and their projection sets \david{the projection of the group itself is not exactly half the size so it does not divide it, that's why we need a coset})
%can be of \emph{any} size $n$ that divides an integer in the Hasse--Weil interval $[q+1 \pm 2\sqrt{q}]$;

%This also yields $O(n \log^2 n)$ size arithmetic circuits over $\fq$ for computing the elementary symetric polynomials of $n$ inputs, when $n = \poly(q)$.

%We show FFT-style divide and conquer algorithms for polynomials using elliptic curve group
%whenever it has a subgroup of large, smooth order.

\end{abstract}

\section{Introduction}

%\swastik{Add other areas of CS where EC helps}
%\swastik{I added the basic skeleton of this, if someone (Eli?) can fill in famous EC crypto references that would complete it.}
%\eli{done, and I think it looks fine}

%\swastik{Do our arithmetic circuits use division? I think they do. We should note this somewhere.}
%\dan{This is a good question. I think the answer is actually no -- I think all our divisions
%are by the precomputed parameters, not by inputs, e.g.\@ when we divided by $Z_0$ in $\REDC$
%or by $X^{n/2}$ in $\EXIT$, so we should simply say that we precompute the inverse and multiply by it,
%instead of division. But we should check no other division is done.}

The rocket fuel that powers modern fast algorithms for polynomial algebra
is the Fast Fourier Transform (FFT). The original FFT, due to Cooley--Tukey~\cite{CooleyTukey}\footnote{The history of this algorithm is much longer, and dates back to Gauss, see~\cite{FFThistory}.}, is a divide-and-conquer algorithm   that evaluates a polynomial $P(X)=\sum_{i<n} a_i X^i \in \mathbb C[X]$,  given by its sequence of coefficient $(a_0,\ldots, a_{n-1})$, on the $n$th roots of unity in $\mathbb C$. It does so using $O(n \log n)$ arithmetic operations over $\mathbb C$ whenever $n$ is an integer power of $2$,
or more generally, when $n$ is a {\em smooth number} -- a product of $O(1)$-sized primes.
%
%The connection to polynomials comes from the observation that the FFT equivalently is an $O(n \log n)$-time algorithm for evaluating a
%polynomial $P(X) \in \mathbb C[X]$ with degree $ < n$, specified by coefficients of its monomials, at the $n$th roots of unity in $\mathbb C$.
This immediately enables $O(n \log n)$ time multiplication of polynomials of degree $< n/2$ -- by evaluation at the $n$th roots of unity, pointwise multiplication of these evaluations, and then interpolation from the $n$th roots of unity via the \emph{inverse} FFT (iFFT) algorithm. Polynomial multiplication turns out to be the crucial operation for a wide variety of other algorithmic problems of polynomial algebra. See the books~\cite{GG-algebrabook,BCS-algebrabook} for a taste of the impact of the FFT on computer algebra.

Over finite fields $\fq$, these ideas generalize to some extent~\cite{pollardFFT}.
%Whenever $\fq$ contains an $n$th root of unity,
%for a smooth number $n$, the FFT yields a fast algorithm for multiplying polynomials of degree
%up to about $n$.
Define $M_q(n)$ to be the number of $\fq$ operations needed for the fastest algorithm over $\fq$ which takes as input the coefficients of two polynomials in $\fq[X]$ of degree $< n$,
and returns the coefficients of their product.
Using the same FFT algorithm, if $\fq$ contains an $n$th root of unity for smooth $n$,
we have $M_q(n) = O( n \log n)$. More generally, we get the same upper bound on $M_q(n)$ even if a bounded degree extension field $\F_{q^{O(1)}}$ contains
such a root of unity which generates a multiplicative subgroup of smooth order.
%, and even if $\F_{q}$ contains
%an additive subgroup of smooth order, which holds for
%all $q$ that are powers of a constant prime; this latter result was shown by Lin, Chung and Han~\cite{LWH2014}, improving on previous results by Cantor~\cite{Cantor} and Gao and Mateer\cite{GM2010}  \swastik{Check!}.
However, most finite fields are not ``special'' in this way, which raises the following well-known open problem:

\begin{center}
\emph{Open Question 1:} \label{quest:quasilinear}
	Does the bound  $M_q(n) = O(n \log n)$ hold for all prime powers $q$ and all $n$?
\end{center}

Until recently, the best general upper bound on $M_q(n)$ was the classical result of
Sch{\"o}nhage and Strassen~\cite{ScSt71}
(see also Sch{\"o}nhage~\cite{Schonhage-char2} and Cantor--Kaltofen~\cite{CaKa}), who showed
that:
$$M_q(n) = O( n \log n \log \log n).$$
This algorithm involves introducing a synthetic root of unity and recursively
running FFTs over more general rings. The algorithm is inspired by, and closely mirrors, the classical (Boolean) algorithm of Sch\"{o}nhage and Strassen for integer multiplication,
which shows that $M_{\Z}(n)$, the Boolean circuit complexity of multiplying two $n$-bit integers presented in base $2$,
satisfies:
$$M_\Z(n) \leq O(n \log n \log\log n).$$

\begin{remark}[Computational Model]
	Unless explicitly specified otherwise, we use the word ``algorithm'' to mean an algebraic algorithm that uses only field operations and field constants. In particular, we do not consider any precision issues or the cost of computing the constants used by the computation. This computational model is more commonly known as an arithmetic circuit or a straight-line program. When we refer to the running time of such an algorithm, we mean the size of the straight-line program or arithmetic circuit, which means we assign unit computational cost to each arithmetic operation over the ambient field.
\end{remark}

As in the case of $\mathbb C$, the best known algorithms for a wide variety of
algorithmic problems of polynomial algebra over $\fq$ depend on polynomial multiplication over
$\fq$, and thus their running time depends on $M_q(n)$.
Of particular interest are the following classical results.
\begin{enumerate}
\item Horowitz~\cite{Horowitz71,Horowitz71Errata} gave an algorithm for polynomial interpolation at $n$ points {\em with preprocessing} in time $O(M_q(n) \log^2 n)$. Here we are given a subset $B$ of $\fq$ of size $n$ and a function $f: B \to \fq$, and after doing arbitrary preprocessing of $B$, we want to compute the coefficients of the unique polynomial of degree $< n$ that interpolates $f$.
\item In the above mentioned paper, Horowitz~\cite{Horowitz71} presented a fast algorithm for evaluating all elementary symmetric polynomials over $n$ variables on a specific input $(\alpha_1,\ldots, \alpha_n)$ in time $O(M_q(n) \log n)$.
\item Subsequently, Borodin and Moenck~\cite{BorodinMoenck} improved Horowitz's algorithm and gave an algorithm for polynomial interpolation at $n$ points {\em without preprocessing} in time $O(M_q(n) \log n)$.
\item Along the way, Borodin and Moenck~\cite{BorodinMoenck} also showed how to do multi-point evaluation of degree $< n$ polynomials at $n$ arbitrary points in time $O(M_q(n) \log n)$.
\end{enumerate}

In recent years, there have been some remarkable advances in our understanding of the complexity of multiplying polynomials over finite fields.
These advances closely track breakthroughs on the fundamental problem of understanding the complexity of multiplying integers in the Boolean circuit or (multi-tape) Turing Machine model.
The starting point for all these recent advances was the result of F\"{u}rer~\cite{Furer} (see also~\cite{Deetal})
who showed that $M_\Z(n) = O( n \log n \cdot 2^{O(\log^* n)})$.
Soon after, Harvey, van der Hoeven and Lecerf \cite{HHL17} simplified and improved the constant in the exponent in
F\"{u}rer's bound on $M_\Z(n)$, while also developing an $\fq$-analogue of this
algorithm to show that $M_q(n) = O(n \log n \cdot 2^{O(\log^* n)})$. Harvey and van der Hoeven \cite{HH19} further improved the constant in the exponent in the bound on $M_q(n)$.

Finally, Harvey and van der Hoeven \cite{HH21} proved the breakthrough $M_\Z(n) = O(n \log n)$, settling a long-standing conjecture. There they discussed the reasons why their results do not extend to a similar bound on $M_q(n)$. Nevertheless, their results do imply (via Kronecker substitutions, see Section 1.2 of~\cite{HH19}) that multiplication of degree $n$ polynomials over $\fq$ for $n = q^{O(1)}$, can be done in time $O(n \log q (\log n + \log\log q))$ in the Turing machine model, which seems to be as good a bound one can hope to deduce in the Turing Machine model from the conjectured $M_q(n) = O(n \log n)$.

Returning to $M_q(n)$, Harvey and van der Hoeven showed in \cite[Theorem 9.2]{HH19linnik}, which is a companion paper to \cite{HH21},  that under a
number theoretic conjecture on the least prime in arithmetic progressions, $M_q(n)$ is indeed $O( n \log n)$.

Summarizing, the recent wave of results come extremely close to answering Open Question 1 unconditionally, but we are not quite there yet.

%\swastik{Also mention the subspace-polynomial-based basis for polynomials in characteristic 2 somewhere.} \eli{done}

\subsection{Our Results}

The main contribution of our paper is a new approach to fast polynomial algorithms via a new polynomial
representation that works over all large finite fields.  The approach is very closely related to the
classical FFT algorithm, but instead of working with subgroups of $\fq$ of smooth order (be they multiplicative or additive), it works with {\em elliptic curve groups} with large, smooth order subgroups,
{\em which exist for all $\fq$}.

Our approach is unrelated to all the recent results mentioned above, and unconditionally yields some
new results that would follow if $M_q(n) = O(n \log n)$ were true.

The new representation for polynomials suggested here is essentially the evaluation tables of the
polynomials at carefully chosen subsets % $S$
% \david{$S$ is incosistent with $L$ above}
% \dan{I think this is fine. The representations are indeed not over all of $L$, but over subsets
% $S \subset L$ which are basic sets}
of $\fq$. These sets % $S$
are related to some subgroup of some elliptic curve over $\fq$.
This is the analogue of representing polynomials by evaluations at % $S$ to be
multiplicative/additive subgroups of large, smooth order, which is only possible when $q$ is
special---either a power of a constant prime or such that $q-1$ is divisible by a large smooth factor.

In the classical multiplicative subgroup based FFT, we can convert the evaluation table representation into the
standard coefficient representation in time $O(n \log n)$ via the classical inverse FFT.
Unfortunately, in our elliptic curve group case, we do not know how to do this conversion as fast.
What we can do instead is to quickly {\em extend} the evaluation of the polynomial on our chosen subset $S$ to
another subset $S'$ of $\fq$. This is the analogue of using a combination of FFT and inverse-FFT
(with some scaling) to use the evaluations of some low degree polynomial at a multiplicative
subgroup $S$ to deduce the evaluations of that low degree polynomial at some coset of $S$. In fact, the way we
compute the low degree extension to the subset $S'$ is also a combination of some FFT-like
transform (which we call the ECFFT) and the inverse transform.
%,where FFT-like means it uses butterfly diagrams, as do classical FFTs.
It just so happens that the intermediate representation, i.e., the result of our iFFT-analogue, is not the standard monomial expansion of the polynomial, but some other representation. In this respect, our approach resembles the additive FFT-like transforms of \cite{LWH2014} which also lead to non-monomial representations supporting fast operations (see also~\cite{GM2010,Cantor}); however, their algorithms have $S,S'$ being additive subgroups of $\fq$, and require $\fq$ to have constant characteristic to have $O(n \log n)$ running time.
% \swastik{Check!}
% \eli{Not sure what needs checking but last sentence lgtm}
% \swastik{Okay, I checked it, we can leave it as is.
%
% Just for future reference for us: The precise results about constant characteristic are sort of messy. I think it is not known whether in the standard monomial representation polynomials can be multiplied in $O(n \log n)$ time. The fastest has $O(n \log n \log\log n)$ time, that too only on certain extension fields $F_{2^r}$ ($r$ has to be a power of 2). Over arbitrary $r$, I believe the fastest known algorithm for polynomial multiplication in the monomial basis is $O(n \log^{(1+\epsilon)}(n))$. Not sure if the integer multiplication breakthroughs changed this.
% }

% \swastik{added footnote about precomputation, ptal}
% \dan{looks great!}
We systematically exploit the above-mentioned fast algorithm for extending
polynomial evaluations on special sets to
develop fast algorithms\footnote{We remind the reader that the model of computation is algebraic circuits (and for one problem, algebraic decision trees), where the circuit may depend arbitrarily on $n$ and $q$. The preprocessing cost of setting up this circuit for a given $n$ or $q$, which in our case involves searching for a suitable elliptic curve, is not included in the complexity bounds. Under standard number theoretic heuristics, this preprocessing can be done by a randomized Turing machine in $O( n \cdot  \poly(\log n, \log q))$ time. Details will appear in~\cite{ECFFT2}.} for a variety of polynomial computation problems, giving the following results, defined formally in \cref{sec:algorithms}:
\begin{enumerate}
\item When polynomials of degree less than $n$ over $\fq$, $n \leq q^{O(1)}$, are represented as evaluations over special sets, the following three operations can all be done in time $O(n \log n)$:
\begin{enumerate}
\item addition,
\item multiplication, and
\item degree computation\footnote{The formal model for this is Algebraic Decision Tree (since the output is an integer), and by ``running time'' for this model we mean the depth of this tree.}
\end{enumerate}
Note that addition trivially takes $O(n)$ time for polynomials evaluated on any set of points, as does
multiplication of polynomials whose degrees sum to less than $n$; the crux here is that polynomials
can still be multiplied in quasi-linear time even if their product has degree above $n$,
by extending the evaluations to a larger set, supporting higher degrees. Degree compuation is also non-trivial,
as the polynomials are not represented directly by their coefficients.
% \dan{I strongly disagree: while addition does indeed take just linear time, this isn't true for multiplication!
% It is not enough to just take the pointwise products, since this does not necessarily lead to a
% representation of the product, whose degree can be larger than the special evaluation set, and thus require
% us to extend it. The ability to extend the special set quickly (which enables multiplication)
% is the crux, not degree computation! I don't think we actually have any real application for degree
% computation... certainly not in this paper.}
% \eli{agree, changed, is it now ok?}
% \david{I don't see that it is changed}
%\dan{It was changed, but not far enough to my taste; I changed it again.}

As far as we know, this is the only known representation of polynomials that allows all the above three operations
to be computed in $O(n \log n)$ algebraic operations for general $q$ and $n \leq q^{O(1)}$.

A folklore question, which was recently resolved
by the breakthrough on integer multiplication~\cite{HH21},
asked to find a representation of integers that supports addition, multiplication and comparison in $O(n \log n)$ time. Our result can be viewed as a positive answer
to the analogous question for polynomials over arbitrary finite fields.

\item We develop fast algorithms for other basic operations on these representations, such as division with remainder and Chinese remaindering, modulu fixed polynomials.
%\dan{Suggest changing ``with preprocessing'' here to the more explicit ``modulo fixed polynomials''.}
%\eli{done}

\item Converting between our new representation and the standard representation of polynomials by their monomial coefficients (in both directions) can be done in time $O(n \log^2 n)$.
\end{enumerate}

Armed with these tools for working with polynomials in the new representation, we get the following new results for
classical problems that have nothing to do with the new representation. All these results improve on the state of the art by a multiplicative $\exp(\log^* n)$ factor, and are consequences of the conjectured bound $M_q(n) = O(n \log n)$. See Section~\ref{sec:classical} for the formal statements.
\begin{enumerate}
\item We give an $O(n \log^2 n)$ time algorithm to evaluate all $n$ elementary symmetric polynomials on $n$ inputs,  provided $n \leq q^{O(1)}$. It was not known how to do this in general for all $n \leq q^{O(1)}$ even for the computation of just the $n/2$-th elementary symmetric polynomial.
\item Given an arbitrary set $B$ of $n$ points, we give an $O(n \log^2 n)$ time algorithm for interpolating a polynomial (and representing it in the standard monomial basis) from its evaluation on $B$ (we allow preprocessing based on $B$).
\item We give an $O(n \log^2 n)$ time algorithm for multi-point evaluation of a degree $<n$ polynomial at an arbitrary set $B$ of $n$ points (here, too, we allow  preprocessing based on $B$).
\item Combining the above two results, we get a an $O(n \log^2 n)$ time algorithm for computing low-degree extensions of function evaluated at $n$ arbitrary points to $n$ other arbitrary points. The two sets of points are assumed to be known in advance, and preprocessed to derive constants used by the algorithm.
\end{enumerate}

We believe this representation will have further uses in the development of fast algorithms for polynomial algebra. The most compelling question here is whether these methods can improve the bound on $M_q(n)$ itself. It is also
% extremely
interesting to see if we can do away with the need for preprocessing in the above algorithms.

\paragraph{Further applications in Part II~\cite{ECFFT2}:}
The applications of FFT-like divide and conquer for polynomials is not limited to the design of
fast algorithms. In a sequel to this paper (which is oriented towards applied cryptography),
we explore applications
of the Elliptic Curve based Fast Fourier Transforms to interactive oracle proofs (IOPs), IOPs of proximity (IOPPs)
for algebraic geometry codes and scalable transparent arguments of knowledge (STARK) systems,
generalizing the use of the standard FFT in PCPPs for Reed--Solomon codes~\cite{BenSassonSudan}, the FRI protocol for
proving proximity to Reed--Solomon codes~\cite{BBHR-FRI}, and the STARK protocol and analogous transparent IOP based
proof systems for verifying general computation~\cite{BBHR-STARK,aurora,fractal,ethSTARK}. Because of applications of
the latter two to cryptography in the real world, where the natural field of definition of the problems is specified by
external sources, there is a natural need to prove computational integrity statements about computations of length $n$
executed over specific finite fields $q\gg n$. For example, the $q$ used in the ECDSA algorithm that is part of the
Bitcoin standard is such that $q-1$ has no large smooth factor, and this is also the case for any $q$ which is a
``safe prime''
%\david{is it realy the condition for a safe prime? I don't think that secp256k1 satisfies it}
which means that $(q-1)/2$ is a large prime. Indeed, such examples were the original motivation for looking for generalizations of FFTs to all fields, and resolving it requires a deeper scrutiny of the ECFFT, used here only as a ``black-box'', and several other ideas.

%\swastik{We can do Ben-Sasson Sudan PCPPS for RS codes over all fields, right? If so, that speaks to cold-hearted
% FOCS reviewers and is worth mentioning.}
%\eli{Right!! You can show the RS-PCPP for RS codes evaluated over the new sets. But you can't get BS PCPs, because
% for arithmetization you need the ``next'' operation, i.e., to use automorphisms of the RS code, and there ain't none
% here (but the AG codes of next paper have them)... Come to think of it, would be best to mention the PCPP next to the
% FRI, they're really quite the same, so, that's for next paper}
%\swastik{Maybe we mention something about studying the FFT transform itself in part 2, not just the extension map.}
%\eli{added, at very end, an appetizer/teaser}
%
%
%\eli{Swastik, I've commented out the ``other related work'' because I don't see why we need to say that we don't get
% anything in the Turing machine model. The main reason is that we haven't really thought about it, so why preclude we
% can't get something better than what follows from the new integer mult stuff? If you want to leave it, I'd suggest
% phrasing it as an open question, i.e., asking whether our results can be used to obtain the same, or improve on, the
% int. mult. stuff. Notice that binary fields results are already discussed.}

\nocite{Pospelov11}

%
%
%\swastik{
%\subsection{Other Related Work}

%We remark that the results of \cite{HH21} on integer multiplication imply (via Kronecker substitution, see Section 1.2 of~\cite{HH19}) that multiplication of degree $n$ polynomials over $\fq$ for $n = q^{O(1)}$, can be done in time $O(n \log q \log(n \log q))$ in the Turing machine model, which seems to be as good a bound one can hope to deduce from the conjectured $M_q(n) = O(n \log n)$. Thus our results do not imply anything new for the
%Boolean circuit/Turing Machine models.
%}
%Mention characteristic 2 guys.

\subsection{ECFFT -- Informal Explanation}
\label{sec:intuition}

%\swastik{June 2, 8:30pm Israeli time - Only this section of the intro remains.}

The standard FFT algorithm exploits the structure of the group of $2^\rounds$-th roots of
unity and its subgroups, using the squaring map  $x \mapsto x^2$ to simultaneously (i) project the
group of size $n$ to a subgroup of half the size and (ii) split a polynomial of degree $n$ into two
polynomials of half the degree, expressed using the squaring map.

Let $n = 2^\rounds$, and suppose we are working in a field $\F$
which contains all $n$ of the $n$th roots of $1$. Let $L\zr \subseteq \F$
denote all the $n$th roots of $1$, assuming we wish
to represent polynomials of degree $<n$
by evaluating them on $L\zr$.
Let $\psi(X) = X^2$ be the squaring map. For each $i$, let
$L^{(i+1)} = \psi(L\ii)$. Thus $L\ii$ is the set of $\frac{n}{2^i}$th roots of
unity in $\F$ and $\psi$ is a $2$-to-$1$ map of degree $2$ from $L\ii$ onto $L\iip$. Thus far we have
described how $\psi$ is used to ``compress'' an evaluation set $L\ii$ to a smaller evaluation set $L\iip$
of half the size. Simultaneously, $\psi$ can be used to ``split'' a polynomial presented in the standard
monomial basis thus:
$$P(X)=\sum_{i<n} a_i\cdot X^i =
\parens*{\sum_{i<n/2} a_{2i}\cdot \psi(X)^i} + X\cdot \parens*{\sum_{i<n/2} a_{2i+1}\cdot \psi(X)^i}
= P_0(\psi(X))+X\cdot P_1(\psi(X)).$$

The FFT evaluates $P$ on $L\zr$ by recursively evaluating both $P_0(Y)$ and $P_1(Y)$ on $y\in\psi(L\zr)=L\one$
and then combining the results using $O(n)$ operations via the formula above. The running time $F(n)$
satisfies the recursive formula $F(n)=2\cdot F(n/2)+O(n)$ leading to $O(n \log n)$ running time.

The essential elements we preserve in our ECFFT are the usage of degree-$2$ maps $\psi\ii$ that are $2$-to-$1$ maps on special sets of points $L\ii$ of size $\frac{n}{2^i}$, along with the ability to express a polynomial
$P(X)$ of degree $<n$  in terms of two other polynomials $P_0(\psi\ii(X)), P_1(\psi\ii(X))$ of degree $<n/2$,
such that the value of $P(x), x\in L\ii$ can be obtained ``locally'' from the values of $P_0(\psi\ii(x)),
P_1(\psi\ii(x))$. Thus, we use such maps and sets of points to describe new \emph{\FFTrees}. An \FFTree\ (see
\cref{def:fftree}) is an ``FFT-inspired'' object that is a layered binary tree whose nodes residing at the
$i$th layer are labeled by the members of $L\ii$, and such that the $2$-to-$1$ map $\psi\ii$ defines directed
edges from two elements $s_0,s_1\in L\ii$ to $t=\psi\ii(s_0)=\psi\ii(s_1)\in L\iip$.

So far we have listed similarities between the FFT and our new ECFFT, so let us now describe the differences.
First, our set $L\ii$ is not a multiplicative group, and in fact it is not a group at all (soon, in
\cref{sec:ode to EC}, we will explain what $L\ii$ actually is). But examining the classical FFT, we could do
its first step using \emph{any} degree-$2$ polynomial $\psi(X)$ which is $2$-to-$1$ on some set of points
$L\zr$ (mapping it to an arbitrary set of points $L\one$ of size $n/2$). The group structure is useful for
knowing, recursively, that we can find further $2$-to-$1$ maps from $L\one$ to $L\two$ and so on.
A second
point of difference is that our $2$-to-$1$ maps may vary with $i$, whereas the classical FFT uses only squaring\footnote{When $n$ is factored into different prime factors (say, $n=2^a\cdot 3^b$)
	one would also use different maps in the FFT (say, squaring and cubing) to move between $L\ii$ and $L\iip$, and varying maps are also used in additive FFTs \cite{GM2010,Cantor,LWH2014}.
%	\david{Additive FFT uses varying maps as well.}
%	\eli{added}
}
to move from $L\ii$ to $L\iip$.
Finally, the maps $\psi\ii$ we use are not degree-$2$ polynomials but rather degree-$2$ rational maps, ratios
of two degree-$2$ polynomials. We show that any such map is just as good for the purpose of splitting a
polynomial into two subpolynomials of half the degree (see \cref{lem:decomposition}), and using rational
maps rather than polynomials gives us more degrees of freedom when  searching for $2$-to-$1$ maps on
special sets of points. These points, and the way they are obtained, are our next, and main, point
in this intuitive description of the ECFFT.

\subsection{Elliptic Curves as a Source for \FFTrees over Arbitrary Finite Fields}
\label{sec:ode to EC}

Elliptic curves are a vast topic of study, with wide-ranging impact across mathematics (e.g.,~\cite{FLT}), and we shall not attempt to describe their importance here.
%\swastik{I think it is a funner reference if we do not say Fermat's last theorem, and just let the reader check the reference to see what we pointed to. So let us make it ``.. across mathematics (for example~\cite{FLT}), and we ... ''. But as is is fine too.}
%\david{I agree with Swastik, nevertheless I think that we should add more references because just one reference feels like not as wide-ranging impact as we claim }
%\dan{fully agree with David}
%\eli{changed phrasing, added one more, please add more}
%\swastik{I think it is above our pay grade to give serious references. We
%do say ``e.g.'' when we say FLT, which is a famous example, and beyond that I think we have no idea of how the landscape and relative bigness of applications looks .. Maybe we can make it ``.. with wide-ranging impact (allegedly) across ...''    :)   }
%\eli{I think the humor of adding (alegedly) won't be appreciated by our dry audience, so leaving only FLT}
An elliptic curve $E$ over the finite field $\fq$ is defined by a suitable polynomial $C(X,Y)\in \fq[X,Y]$,
and the solutions $(x,y)\in \fq^2$ of $C(X,Y) = 0$ are the points of interest (the description here is
intentionally simplified, see \cref{sec:EC background} for a formal and accurate definition). Elliptic
curves have some remarkable properties that have led to a number of significant and surprising applications
in theoretical computer science. A small selection of notable examples include:
Lenstra's elliptic curve method for factoring integers~\cite{Lenstra};
Schoof's deterministic algorithm for finding square roots modulo a prime~\cite{Schoof};
Couveignes and Lercier's randomized algorithm for finding irreducible polynomials over finite
fields~\cite{CoLe};
cryptosystems, starting with Miller's EC Diffie--Hellman (ECDH) key exchange~\cite{ECDH}
and Koblitz's EC integrated encryption scheme (ECIES)~\cite{ECEG, ECAES} and including Vanstone's EC digital signature algorithm (ECDSA)~\cite{ECDSA}
and applications based on pairings, such as Joux's one-round 3-way key agreement~\cite{joux} and the Boneh--Franklin identity based encryption protocol~\cite{BonehF03}.
%\dan{I'm not sure it makes sense to explicitly credit Koblitz with ECIES specifically, it seems he was the first
%to suggest EC ElGamal, which is a part of ECIES, but that \cite{ECAES} has more than that.}
%\eli{looks good}

We remark that Lenstra's method for factoring integers using elliptic curves \cite{Lenstra} in
particular was a major inspiration for this paper. Lenstra's method is a generalization of
Pollard's $p-1$ algorithm
for factoring~\cite{p-1}: The $p-1$ method only works when, for some prime factor $p$, the
multiplicative group $\F_p^\times$ has a special property, which is only true for few primes $p$.
Lenstra's method extends the $p-1$ method to all possible $p$'s by replacing the group $\F_p^\times$
with elliptic curves. Very similarly, the standard FFT works inside $\fq$ only when the
field has special roots of unity, which is true only for sporadic $q$, and this paper extends core
applications of FFT to all prime powers $q$ by replacing the group $\fq^\times$ with elliptic curves.

%(such as ... cite the first paper to propose elliptic curve groups for discrete log, the first bilinear pairings paper, maybe other famous EC-crypto firsts). \swastik{NEED TO FILL IN THE ABOVE.} \eli{done}

The main properties of elliptic curves that we use are:
%\dan{Maybe change both ranges below to $[q \pm 2\sqrt{q} + 1]$, as in the abstract? i.e.\@ keep the
%big-O notation only for $n=O(\sqrt{q})$.} \eli{accepted}
\begin{itemize}
    \item The number of points on the curve $E$ can be nearly any number in the range $[q + 1 \pm 2\sqrt{q}]$ (this is by theorems of Deuring and Waterhouse; see \cref{sec:group struct} for a precise discussion of the number of points).
	\item These points
	%(alongside a special non-$\fq$ point)
	%\david{the special point should appear earliear because we already talked about the number of points in the previous item}
%	\dan{removed the parenthetical entirely, no real need to mention the special point here}
	form an abelian group, called, appropriately, an \emph{elliptic curve group}. Varying over curves, and acknowledging  the previous point, elliptic curve groups could be of nearly any size in $[q\pm 2\sqrt{q}+1]$. In particular, we can find subgroups $G$ of elliptic curve groups of size $n=2^k$ for $n=O(\sqrt{q})$ (see \cref{thm:Waterhouse,cor:Waterhouse}).
	\item If $H < G$ are subgroups of an elliptic curve $E$  over $\fq$, there is an $|H|$-to-$1$ map $\phi$ (called an isogeny) with kernel $H$ from the points of the curve $E$ to points on a different curve $E'$ over $\fq$. Thus, the image of $G$ under the isogeny is of size $|G|/|H|$.
	%Additionally, this map is given by a pair of rational functions $\phi_x, \phi_y$ and $\phi_x$ is a degree-$|H|$ rational function (see \cref{prop:quotient isogeny}).
\end{itemize}

The observations above give us nearly all that we need. We can find a set of points $G\zr$ inside a curve $E\zr$ that
is a group of size $2^\rounds\leq O(\sqrt{q})$ irrespective of the exact nature of $q$, and we have at our disposal
isogenies  that ``compress'' groups of points $G\ii$ to groups $G\iip$ half the size via $2$-to-$1$ maps $\phi\ii$,
where the new group $G\iip$ belongs to a different curve $E\iip$. The only remaining gap is that elements in the groups
$G\ii$ are \emph{pairs} $(x,y)\in \fq^2$ whereas we are interested in univariate polynomials and evaluation sets over
$\fq$. The final ingredient is to pick curves represented in a certain format (extended Weierstrass form) such that
suitably shifting and then projecting $G\ii$ to the $x$ coordinate gives a set $L\ii\subset \fq$
that is the same size as $G\ii$ and, crucially, the
isogeny map $\phi\ii$ gives rise to a degree-$2$ rational map that is $2$-to-$1$ from $L\ii$ onto $L\iip$
(see \cref{prop:standard form,thm:curve sequence}).

\begin{remark} The degree-$2$ (or higher degree) maps so obtained are generalizations of {\em Latt\'es maps}~\cite{Latte} (see~\cite{Silverman-dynamics}).
Latt\'es maps are the rational maps arising from the $x$-coordinate mapping of isogenies from an elliptic curve to {\em itself}. The 
rational maps that underlie the $\FFTree$ construction arise from the $x$-coordinate mapping of isogenies from an elliptic curve $E$ to some other elliptic curve $E'$, which may or may not equal $E$.
\end{remark}

Summarizing, the abundance of elliptic curve groups of various sizes over any large finite field assures us that
we will find a subgroup of smooth size; isogenies and their projections give $2$-to-$1$ degree-$2$ rational maps from sets of size $2^\rounds$
(in $\fq$) to sets of size $2^{\rounds-1}$ for all needed $\rounds$, and thereby we have the needed \FFTree structure which 
leads to efficient FFT-like running times for all finite fields.

%
%
%Now we bring in the main ingredient: elliptic curves. Elliptic curves over finite fields
%are algebraic curves that come with an associated group operation defined by
%rational maps. There is a huge world of elliptic curves and group homomorphisms (given by rational maps)
%between them, and it is here where we find the desired FFT-friendly structures.
%We focus on elliptic curves with subgroups of size $2^k$, ..

\paragraph{The paper of Chudnovsky and Chudnovsky~\cite{ChCh}}
Some of the core ideas appearing in this paper were first suggested, in brief,
in the final section of a paper by Chudnovsky and Chudnovsky \cite[Section 6]{ChCh}.
The main claim from~\cite{ChCh} relevant for us is an ``Elliptic Interpolation Algorithm'' called 
FENTT (Fast Elliptic Number Theoretic Transform), which
describes how to use elliptic curve groups over finite fields
to solve a certain rational function interpolation problem via
an FFT-type algorithm.

We were made aware of this paper by an anonymous reviewer,
who further remarked that the authors only sketch
their idea. Indeed, the writing is extremely succinct, and many details are omitted or only
hinted at. Furthermore, we believe there is at least one critical point which was overlooked in that
work, which significantly limits the applicability of the FENTT. 

As a consequence of this error, it turns out (via a result published independently by Ruck~\cite{Ruck}
and Voloch~\cite{Voloch})
that the instance size $n$ of an FENTT over $\F_q$ is bounded from above in terms of the prime
factorization of $q-1$ (as is the case for FFTs).
The net result is that an FENTT of size $n$ can be computed over $\F_q$ only if 
an FFT of size $\Theta(\sqrt{n})$ can be computed over $\F_q$ --- thus, finite fields
that do not support large FFTs also do not support large FENTTs. In contrast, large ECFFTs 
(of size $q^{\Omega(1)}$) are supported by all finite fields $\F_q$. This fact was crucial to our
faster algorithms for working with polynomials of degree $n$ over $\F_q$ for all $n \leq q^{O(1)}$.

We give details about the ideas, methods, and limitations of~\cite{ChCh} in relation to ours in Appendix~\ref{sec:ChCh-comp}.

\paragraph{Organization of paper}
The following \cref{sec:notation} gives notation. \cref{sec:FFTrees} defines and discusses (i) the \FFTree data structure and (ii) the polynomial decomposition lemma (using rational maps); these two ingredients are needed to abstract and generalize
the classical FFT algorithm to arbitrary sets of points and maps.
\cref{sec:FFTrees from EC} instantiates
\FFTrees and decomposition maps using elliptic curve and projections of isogenies, showing
that the necessary data structures exist over all large finite fields. \cref{sec:representations} defines the way we represent polynomials for efficient operations -- by evaluating them over the
special sets of points that arise from the previously defined \FFTrees.
\cref{sec:algorithms} presents fast algorithms for
fundamental operations applied to polynomials that are represented in this special way.
Finally, \cref{sec:classical} uses these efficient algorithms to efficiently solve ``classical'' problems about polynomials, like interpolation, evaluation over general sets of points, and computation
of elementary symmetric polynomials.

\iffalse{
Our FFT-like polynomial algorithms are based on a number of variations on the above structure.

First, we allow the $2$-to-$1$ maps between the consecutive $L\ii$ to be
different maps $\psi\ii: L\ii \to L\iip$. It is easy to see that the
decomposition above holds with any quadratic polynomial of $X$ in place of $X^2$.

Even with this flexibility, it is not clear that what we want exists. Is it true
that for every finite field there a sequence of sets $L\ii$ and
quadratic polynomials $\psi\ii$
We do not know.

Next we allow the $\psi\ii$ to be {\em rational maps} of degree $2$. While this introduces
some complications because we move out of the space of polynomials, it turns out that
an analogue of the above decomposition does hold.

With this FFT-friendly family of rational maps in place, we develop data structures and algorithms
for working with polynomials that exploit it. Here a new challenge arises.
When we worked with the simple map $\psi(X) = X^2$.
}\fi

\section{Notation}
\label{sec:notation}

%\subsection{Paper-specific notation}

\subsection{Functions and Polynomials}

For $g:D\to R$ a function and $S\subset R$ denote by $g^{-1}(S)$ the set of $g$-preimages of $S$,
%\david{did you mean preimages of $S$?} \dan{fixed}
namely $g^{-1}(S)=\{x: g(x)\in S\}$, and for $u\in R$ let $g^{-1}(u)=g^{-1}(\{u\})$. Likewise for $D'\subset D$ we let $g(D')=\{g(x): x\in D'\}$.

For a set $A \subseteq \fq$, we define {\em the vanishing polynomial} of $A$
to be the polynomial $Z(X) \in \fq[X]$ given by:
$$ Z(X) = \prod_{\alpha \in A} (X - \alpha).$$

We define:
$$ B(X) \rem A(X)$$
to be the unique polynomial with degree $ < \deg(A)$
which is congruent to $B(X)$ mod $A(X)$.

When $B(X), A(X$) are coprime polynomials, we define
$$(B(X))^{-1}_{A(X)}$$
to be the unique polynomial $C(X)$ with degree $< \deg(A)$
such that $B(X) \cdot C(X) \equiv 1 \pmod{A(X)}$.

\subsection{Projective Space}
We denote by $\P^n(\fq)$ (or simply $\P^n$) the \emph{$n$-dimensional projective
space} over $\fq$; only $\P^1$ and $\P^2$ will appear in the paper.
Points in $\P^n$ are given by homogenized coordinates $[x_1:x_2:\dots: x_{n+1}]$
where at least one $x_i$ is non-zero, and with the equivalence relation
$$[x_1:x_2:\dots: x_{n+1}] \sim [cx_1 : cx_2 :\dots :cx_{n+1}], \quad \forall c \neq 0.$$

Points in the \emph{affine} space $\fq^n$ are given by affine coordinates $(x_1,\dots, x_n)$,
and in this paper we equate such points with their standard embedding into projective space, i.e.\@
$$(x_1, \dots, x_n) = [x_1 :\dots : x_n : 1].$$
Thus, $\P^n$ is the disjoint union of $\fq^n$ and a copy of $\P^{n-1}$ ``at infinity'', i.e.\@ with
an additional $x_{n+1}=0$ coordinate. In particular, $\P^1(\fq) = \fq \cup \{\infty\}$, where
$\infty$ denotes the unique point at infinity, $[1:0]$.

We will refer to the two coordinates of the affine plane $\fq^2$ as $x$ and $y$. For a point
$P \in \fq^2$, we will denote its $x,y$ coordinates by $P_x, P_y$, respectively. For a point
$P \in \P^2$, the coordinates $P_x, P_y$ will only be defined if it is an affine point, according
to the above notation.

\subsection{Rational functions}
\emph{Rational functions} over $\fq$ are quotients $R(X) = P(X)/Q(X)$ where $P(X),Q(X) \in \fq[X]$
are coprime polynomials and $Q$ is non-zero. Rational functions form a field, denoted by
$\fq(X)$.

Rational functions can be considered as maps from $\P^1$
to itself, where zeros of $Q$ are mapped to $\infty$ and are called \emph{poles} of the rational
function, with multiplicity equal to their multiplicity as zeros of $Q$. Depending on whether
$\deg(P) - \deg(Q)$ is positive, negative, or zero, the point $\infty$ is
either a pole of multiplicity $\deg(P) - \deg(Q)$, a zero of multiplicity $\deg(Q) - \deg (P)$, or mapped
to the ratio between the leading coefficients of $P$ and $Q$, correspondingly.

The \emph{degree} of $R$ is defined as $\deg(R) \coloneqq \max(\deg(P), \deg(Q))$, and is equal
to both the total number of zeros and the total number of poles of $R$, including at $\infty$,
counted with multiplicity.

\section{Polynomial decompositions and \FFTrees} \label{sec:FFTrees}

In this section we show that any rational map can be used to decompose a polynomial
into lower degree polynomials, in a way similar to how the squaring map is used in FFTs
(see \cref{lem:decomposition}). We then define a generalized notion of FFT-like sets of evaluation
points (\cref{sec:fftrees definition}). In the next section we shall instantiate both of
these---rational maps and \FFTrees---using elliptic curve groups.

\subsection{Polynomial decompositions based on rational functions}
\label{sec:EC poly decomposition}

Let $V_d$ be the $\fq$-linear subspace of $\fq[X]$ consisting of
polynomials of degree strictly less than $d$. A crucial component in the standard FFT is the decomposition of a polynomial $P(X)=\sum_{i<d} a_i X^i\in V_d$ into two polynomials in $V_{d/2}$, one containing the terms of even degree and the other containing the terms of odd degree:
\begin{equation}\label{eq:FFT equality}
	P(X)=\left(\sum_{i<d/2} a_{2i} \left(X^2\right)^i\right)+ X\cdot \left(\sum_{i<d/2} a_{2i+1} \cdot \left(X^2\right)^i\right) = P_0(X^2)+X\cdot P_1(X^2).
\end{equation}
The results of this section generalize this partition by replacing $X^2$ with any rational function. Later, we shall instantiate the results of this section with rational functions coming from projections of isogenies of elliptic curves. We state the decomposition lemma next; its proof appears
in Appendix~\ref{sec:proof-of-decomposition}.

\begin{restatable}[Decomposition]{lemma}{decomp}
\label{lem:decomposition}
Let $\psi(X) \in \fq(X)$ be a rational map given by:
$$ \psi(X) = \frac{u(X)}{v(X)},$$
where $u(X),v(X) \in \fq[X]$ are relatively prime polynomials.
Let $\delta = \deg(\psi) = \max \{ \deg(u), \deg(v) \}$.
Let $d$ be a multiple of $\delta$.
%For a function $f : \plin \to \plin$,
%define $\psi^*(f) = f \circ \psi$.
Then for every $P(X) \in V_d$, there is a unique tuple:
$$ \left(P_0(X), P_1(X), \ldots, P_{\delta-1}(X)\right) \in (V_{d/\delta})^\delta$$
such that:
\begin{equation}\label{eq:ECFFT equality}
	P(X) = \left( \sum_{i=0}^{\delta-1} X^i \cdot P_i(\psi(X)) \right) \cdot v(X)^{\frac{d}{\delta}-1}.
\end{equation}
%\swastik{Seems like we never use this .. also we didn't define $\psi^*$
%Thus:
%\begin{align}
%\label{dsum}
%V_d = \bigoplus_{i=0}^{\delta-1} X^i \cdot \psi^*\left(V_{d/\delta}\right) \cdot v(X)^{\frac{d}{\delta}-1}.
%\end{align}
%}
\end{restatable}

%\eli{suggest adding a part, or separate lemma that deals with locality, and will be used later in the algorithmic part, here's the sketch of it. If we add it, then it can replace the equations (2)-(4) below.

The next statement says that, as in the case of the standard FFT,
moving between the two representations of \cref{eq:FFT equality,eq:ECFFT equality} is done via a set of $\delta$-local invertible linear transformations.

\begin{lemma}[Locality and invertibility]\label{lem:locality}
Let $t \in \fq$. Keeping the notation of the previous lemma,
suppose $\psi^{-1}(t)=\{s_0, \ldots, s_{\delta-1}\}$ is a set of elements of $\fq$ of size exactly $\delta$.
Then the transformation \begin{equation}\label{eq:M} M_{t}:\fq^\delta\to \fq^\delta, \quad M_{t}(P(s_0),\ldots,P(s_{\delta-1}))\mapsto (P_0(t),\ldots, P_{\delta-1}(t))
	\end{equation} is linear and invertible.
\end{lemma}
%\dan{I switchd $i,j$ in the proof of the lemma, so that they correspond to the row/column indices of the
%left multiplication matrix in the standard way. For me this makes the last line more obvious (instead
%of appearing transposed). The new indexing of $P_j$ is maybe incosistent with the above $P_i$; we
%could also change all the $i$s to $j$ but that's probably unnecessary. If the constistency is more
%important than the orientation of the matrix, switch back.}
%\swastik{I switched back for consistency }
\begin{proof}
The assumption $t\in \fq$ and, in particular, $t\neq \infty$, implies $v(s_j)\neq 0$ for each $s_j$.
The relationship between the $P(s_j)$ and the $P_i(t)$ is captured by the following system of linear equations:
$$ P(s_j) = \left( \sum_{i=0}^{\delta-1} s_j^i \cdot P_i(t) \right) \cdot v(s_j)^{\frac{d}{\delta}-1}.$$
Inspection shows that the underlying matrix is a nonsingular Vandermonde matrix with rows scaled by nonzero scalars.
\end{proof}

For the rest of this paper, we will focus on the $\delta = 2$ case, although
everything generalizes to larger $\delta$. We briefly instantiate the above lemmas in this case,
to expose the similarity to the classical FFT.

Let $\psi(X)$ be a degree $2$ rational function. Suppose $d$ is even.
Fix any $P(X) \in V_d$, and consider the two polynomials $P_0(X), P_1(X)$ given by \cref{lem:decomposition}.
Then we have the following decomposition that resembles the classical FFT case of \cref{eq:FFT equality}:
$$P(X) = \left( P_0(\psi(X)) + X P_1(\psi(X)) \right) \cdot (v(X))^{\frac{d}{2}-1},$$
and so, for any $s \in \fq$:
\begin{align}
\label{eqPmain}
P(s) = \left( P_0(\psi(s)) + s P_1(\psi(s))\right) \cdot (v(s))^{\frac{d}{2}-1}.
\end{align}

Let $s_0, s_1, t \in \fq$ be such that $\psi(s_0) = \psi(s_1) = t$ with
$s_0 \neq s_1$. Then Lemma~\ref{lem:locality} implies that the values $P(s_0), P(s_1)$ determine $P_0(t), P_1(t)$
and vice versa (this uses the fact that $s_0 \neq s_1$), and the transformation between the two pairs of values
is computed by multpilication by an invertible $2\times 2$ matrix, whose coefficients depend only on
the values of $s_0, s_1, v(s_0)$, and $v(s_1)$.
% \david{I'm not sure that the term 'fixed' is clear enough here, it is not a constant, but it depends only on $s_i$ and $\psi$}
% \dan{agreed, changed}

Thus, when we have a degree 2 rational function $\psi$ that is $2$-to-$1$ from $S$
% \david{again inconsistency between the notations $S$ and $L$}
% \dan{Again I think this is fine, this claim is relevant to basic sets $S$ of any size and not necessarily all of $L$}
to $\psi(S) = T$, finding evaluations of a polynomial $P(X)$ at the points of $S$ is equivalent to finding
evaluations of $P_0(X)$ and $P_1(X)$ at the points of $T$.

\subsection{\FFTrees}\label{sec:fftrees definition}

We now define \FFTrees, a structure abstracting out relevant properties of evaluation sets and maps between them, which suffice to simulate
an FFT-like algorithm.

\begin{definition}[\FFTrees]
	\label{def:fftree}
Let $q$ be a prime power, and let $\rounds$ be an integer.
An $\FFTree$ over $\fq$ of \emph{depth} $\rounds$
is a collection of subsets $L\zr, L\one \ldots, L\fin \subseteq \fq$ along
with degree $2$ rational functions $\psi\ii(X) \in \fq(X)$ such that:
\begin{enumerate}
\item $|L\ii| = 2^{\rounds - i}$.
\item $\psi\ii(L\ii) = L\iip$ (and so $\psi\ii$ is a $2$-to-$1$ map from $L\ii$ to $L\iip$).
\end{enumerate}
Let $\calF$ denote the rooted, layered, binary tree, whose layers are indexed by $i \in \{0,1,\ldots, \rounds\}$. The set of vertices in layer $i$ is $L\ii$. The root
of $\calF$ is the unique element of $L\fin$.
The leaves of $\calF$ are all the vertices in $L\zr$.
For each $i < \rounds$, the parent of the vertex $s \in L\ii$ of the $i$-th layer is the vertex $\psi\ii(s) \in L\iip$ of the $(i+1)$st layer.
\end{definition}

Because of the decomposition lemma, evaluations of a polynomial on $L\ii$ can be deduced from evaluations of 2 related lower degree polynomials on $L\iip$, and this serves as the basis for fast ``divide and conquer'' algorithms.

Our eventual use of $\FFTrees$ will be as follows. We will first fix an $\FFTree$ over $\fq$. We will use $L$ to denote $L\zr$.
Let $\nn = |L| = 2^{\rounds}$.
Then for any $n \le \nn$, polynomials of degree $< n$
% \david{$n$ or $\nn$?}
will be represented by evaluations at specific subsets of $L$ of size $O(n)$.
% \david{what do you mean by specially chosen points? I suggest: ``Then polynomials of degree $< \nn$ will be represented by evaluations at $L$''}
The $\FFTree$ structure will then enable fast algorithms for working with these representations.

Thus any given $\FFTree$ will be useful for working with polynomials of degree up to $2^{\rounds}-1 $.
Therefore it is interesting to find $\FFTrees$ with as large depth $\rounds$ as possible.

In the next section, we use elliptic curves to show the existence
of $\FFTrees$ over $\fq$ with depth $\Omega(\log q)$.

\section{$\FFTrees$ from Elliptic Curves}
\label{sec:FFTrees from EC}
In this section we prove the existence of $\FFTrees$ of depth $\Omega(\log q)$ in any finite field $\fq$. Specifically, we show that there exist $\FFTrees$ over $\fq$ whose base set $L\zr$ has size $\Omega(\sqrt{q})$.
% Later on we shall use such
%sets of points to obtain fast algorithms for fundamental operations on polynomials.
We start by recounting the necessary definitions and results regarding elliptic curves.
In \cref{sec:EC nice maps} we then prove our main results about existence of $\FFTrees$ using rational
maps that are projections of isogenies.
%In \cref{sec:EC poly decomposition} we prove that polynomials can be ``decomposed'' using any rational map, in particular, the rational maps used to construct our FFT-like sets.

\subsection{Background on elliptic curves and isogenies}\label{sec:EC background}

In this subsection we provide a brief overview of the necessary definitions and theorems regarding
elliptic curves.
Further details and proofs can be found in most basic texts on the subject. Except where
specifically noted, all results can be found in \cite{Silverman} or \cite{Wash08}.

\subsubsection{Elliptic curve in Weierstrass form}
An \emph{elliptic curve} $E$ is a smooth, projective, algebraic curve of genus 1, with a special
marked point $O$, defined over a field. In this paper all curves will be defined over the finite field $\fq$.
Every elliptic curve can be presented in \emph{extended Weierstrass form}
as the set of planar points $(x, y) \in \fq^2$ satisfying a cubic equation
\begin{equation}
	 Y^2 + a_1 XY + a_3 Y = X^3 + a_2 X^2 + a_4 X + a_6
\end{equation}
or equivalently
\begin{equation}
	\label{eq:weierstrass}
	F(X,Y):=
		Y^2 + a_1 XY + a_3 Y - X^3 - a_2 X^2 - a_4 X - a_6 = 0
\end{equation}
parameterized by $a_1,a_2,a_3,a_4,a_6$,
together with the marked point $O = [0:1:0] \in \P^2(\fq)$, called its \emph{point at infinity}.

% For any curve $E$, we denote by $\fq(E)$ the field of rational functions from $E$ to $\P^1(\fq)$, which
% are ratios $P(X,Y)/Q(X,Y)$ of bivariate polynomials over $\fq$, where $Q$ is not identically zero
% on $E$, taken modulo the cubic equation defining $E$.

\subsubsection{The group law}\label{sec:group law}
The points of an elliptic curve $E$ form an abelian group, in which $O$ is the neutral element, and
any three distinct points $P, Q, R \in E$ satisfy $P+Q+R = O$ iff they are colinear.
If $P=Q\neq R$, the condition
is that the tangent to $E$ at $P$ passes through $R$, and if $P=Q=R$ the condition is that the tangent
at $P$ to $E$ is doubly tangent at the point.

Lines passing through $O$ are either the line at
infinity (which is doubly tangent to $E$ at $O$), or lines of the form $X = c$. Thus $P+Q = O$,
i.e.\@ $P = -Q$, iff their coordinates satisfy $P_x = Q_x$ and $P_y \neq Q_y$; or $P=Q\neq O$ and the line $X = P_x$
is tangent to $E$ at $P$; or $P = Q = O$. Note that in both affine cases, %\david{did you mean 3 cases?}
we also have
$Q_y = -a_1P_x -a_3 - P_y$, since $P_y, Q_y$ are the two (not necessarily distinct) roots of
a monic quadratic in $y$ with linear coefficient $a_1 P_x + a_3$.

\subsubsection{Isogenies and $x$-projection}
For a curve $E$ in extended Weierstrass form, let $\pi: E \to \Prj$ denote the projection to the $x$-coordinate,
defined by $\pi(O) = \infty \in \Prj$ and $\pi(P) = P_x \in \fq$ for $P \in E \setminus \{O\}$.
Additionally,
as noted in \cref{sec:group law}, for any $P, Q \in E$, $\pi(P) = \pi(Q)$ if and only if $P = \pm Q$,
thus the preimages $\pi^{-1}(\pi(P)) = \{\pm P\}$ are either sets of size two, or
% , for at most 4 points,
%\eli{pls explain more explicitly why 4} \dan{4, what 4? :)}
a singleton $\{P\}$ when $2P = O$. In particular, it follows that for any subset $C \subset E$
such that $C$ is disjoint from $-C = \{-P : P\in C\}$, the map $\left. \pi \right|_{C}$
is $1$-to-$1$ from $C$ to $\fq$.

% \eli{looking ahead, would be good to have claim that says ``Suppose $S$ is a set of points in $E$ such that $S\cap -S=\emptyset$, then $\pi(S)$ is $1$-to-$1$ onto $\Prj$.'' Then, this claim can be quoted at end of proof of \cref{cor:Lpsi seq}.}

Let $E, E'$ be elliptic curves over the same field. An \emph{isogeny} between the curves is a
% \eli{rational?} map $\phi : E \to E'$, given in coordinates \eli{what does ``given in coordinates''
% mean} by rational functions, and satisfying $\phi(O) = O'$.
rational map $\phi : E \to E'$ satisfying $\phi(O) = O'$, where $O'$ is the neutral element of $E'$.
We follow \cite[Chapters 2.9, 12.2]{Wash08} to give an algebraic, rather than geometric, description
of isogenies. When $E, E'$ are in extended Weierstrass form, $\phi$ can be expressed in a
\emph{standard form}:

\begin{restatable}{proposition}{standard}\label{prop:standard form}
 Let $\phi: E \to E'$ be an isogeny between two curves in extended Weierstrass form. Then, in coordinates,
 we may write
 $$\phi(x, y) = (\psi(x), \xi(x,y)),$$
where $\psi : \Prj \to \Prj$ is a rational function. Equivalently, if $\pi : E \to \Prj, \pi':E'\to \Prj$
are the $x$-projection maps in each curve, then there exists a unique
rational function $\psi$ such that the diagram
\begin{equation*}
 \begin{tikzcd}[ampersand replacement=\&]
    E \arrow[swap]{d}{\pi} \arrow{r}{\phi} \& E' \arrow{d}{\pi'} \\
    \Prj \arrow{r}{\psi} \& \Prj
  \end{tikzcd}
\end{equation*}
is commutative.
% Additionally, we have $\deg \psi = \deg \phi$, i.e.\@ the degree of $\phi$, as
% an isogeny, is equal to the degree of $\psi$ as a rational function.
\end{restatable}

This fact appears to be folklore, and is most commonly discussed only in the special case of curves
in \emph{short} Weierstrass form $E : y^2 = x^3 + Ax +B$, where $\xi(x,y)$ can also be expressed as
$y$ times a rational function---see \cite[Chapter 2.9]{Wash08} for a discussion of this case.
When focusing only on the $x$-coordinate, the same proof is valid also for the extended Weierstrass form.
For completeness, a full proof of this fact is included in \cref{sec:psi appendix}.

\begin{definition}
	\label{def:degree of projection equals degree of isogeny}
 Let $\phi: E \to E'$ be an isogeny between two curves in extended Weierstrass form,
 and let $\psi$ be as in \cref{prop:standard form}. We define $\deg \phi \coloneqq \deg \psi$,
 i.e.\@ the degree of the isogeny $\phi$ is defined to be equal to the degree of $\psi$ as a rational
 function.
 The isogeny $\phi$ is called \emph{separable} if the derivative (in $x$) of $\psi$ is not identically zero.
\end{definition}
The term \emph{$d$-isogeny} is shorthand for degree $d$ isogeny.

An important property of isogenies is that they are also group homomorphisms, with finite kernels.
If $\phi$ is separable, then $|\ker \phi| = \deg \phi$.
The converse is also true, and is a crucial part of our
construction:
\begin{proposition}[{\cite[III.4.12]{Silverman}}]\label{prop:quotient isogeny}
 Let $E$ be an elliptic curve and let $H < E$ be a finite subgroup of $E$. There is a unique
 elliptic curve $E'$ and a separable $|H|$-isogeny $\phi: E \to E'$ with $\ker\phi = H$.
\end{proposition}
See also \cite{Velu71} for an explicit construction of such isogenies. We will apply the proposition
for groups $H$ with $|H| = 2$, but all our results generalize to larger $H$. In this case
$\phi$ is 2-isogeny, meaning  $\psi$ is a degree 2 rational function.

\subsubsection{Group size and structure}\label{sec:group struct}
The group $E$ is abelian, and it is always of rank at most 2, i.e.\@ it is isomorphic to a product of
at most 2 cyclic groups
$$ E \simeq \Z/m_1\Z \times \Z/m_2\Z$$
with $m_1 \mid m_2$ and $m_1 \cdot m_2 = |E|$.

Hasse's theorem states that for every elliptic curve $E$, the order of the group $|E|$
%\eli{$E(\fq)$, and have we defined the set of $\fq$-rational points and mentioned that obviously they're a subgroup?}
%\dan{Actually my notation was $E = E(\fq)$, i.e.\@ $E$ is defined as a set of $\fq$ points,
%and the group is defined directly to be $E$, so this is not a subgroup but \emph{the} group.
%You are right that this is not perfectly standard but I think it's OK, since we don't really need
%field extensions of $E$...}\eli{In the proof of Thm 4.2 the closure group is mentioned}
belongs to a range
of length $4\sqrt{q}$ centered at $q+1$, that is,
$$ q - 2\sqrt{q} + 1 \le |E| \le  q + 2\sqrt{q} + 1.$$
By a theorem of Deuring~\cite{Deur41}, any number in this range is indeed attainable as the size
of an elliptic curve, in the case where $q$ is prime. Waterhouse~\cite[Theorem 4.1]{Wat69} provides
the complete characterization of achievable sizes for the prime power case.
We will require a much weaker form, about possible factors of $|E|$. The following is the
simplest case of Waterhouse's theorem:
\begin{theorem}\label{thm:Waterhouse}
 Let $N = q + 1 - t$ be an integer such that $|t| \le 2\sqrt{q}$ and $t$ is coprime to $q$.
 Then there exists an elliptic curve $E/\fq$ with $|E| = N$.
\end{theorem}

\subsection{An FFT-friendly sequence of rational maps coming from elliptic curves}
\label{sec:EC nice maps}

% \dan{Perhaps we should add subsubsections here.}
% \eli{I'm fine either way. Reading it again, perhaps the needed restructuring would be to start with
% \cref{cor:FFTree-existence}, changing its name to theorem (Main theorem?), then calling
% \cref{cor:Lpsi seq} also a Theorem (Main Technical?) and then proving them. This way we have the main
% claim, that can be used elsewhere, appear upfront. Dan, if you like this suggestion would you like to give it a try?}
% \dan{Gave it a try, PTAL; it can't appear immediately upfront since it needs $\nnq$, and also
% the lower bound on it, I tried to place it as early as possible.}
% \swastik{Looks great to me.}

As noted in \cref{sec:FFTrees}, the depth (or size) of an $\FFTree$ limits the degrees of the
polynomials which it can be used to evaluate. Thus, we would like to find the largest $\FFTree$
possible: if smaller degrees are sufficient, we can always use a subtree instead. We will denote by
$\nnq$ the largest possible size of an $\FFTree$ which can be obtained by our method. More
rigorously, we define
\begin{definition}\label{def:nnq}
 Let $q$ be a prime power. Define $\nnq$ to be the largest power of 2 such that there exists an
 elliptic curve $E$ defined over $\fq$ whose size satisfies $\nnq \mid |E|$ and $|E| > 2\nnq$.
\end{definition}
We claim that $\nnq$ is in fact fairly large with respect to $q$:
% Note that if $q\ge 4$ is even, then $\nnq = \frac{q}{4}$.
% \eli{Why? is $t=1$ considered co-prime to $q$? If yes I think this should be explicitly reminded}
% More generally, we claim:
\begin{claim}\label{cor:Waterhouse}
 Let $q \ge 7$ be a prime power. Then $\nnq > \sqrt{q}$. Equivalently,
 for any $\nn = 2^\rounds \le 2\sqrt{q}$,
 there exists an elliptic curve $E$ defined over $\fq$ with $\nn \mid |E|$ and $|E| > 2\nn$.
 If $q$ is even, then $\nnq \ge \frac{q}{4}$.
\end{claim}

Before we prove \cref{cor:Waterhouse}, having defined and bounded $\nnq$, we are now able to
precisely state the main theorem of this section:

\begin{theorem}[Existence of large $\FFTrees$]\label{cor:FFTree-existence}
Let $q$ be a prime power, and let the integer $\rounds$ be such that $\nn = 2^\rounds \le \nnq$; in particular,
one may take $\nn$ to be any power of two up to $2\sqrt{q}$ for $q \ge 7$.

Then there exists an $\FFTree$ over $\fq$ with depth $\rounds$.
\end{theorem}

We now proceed with building up the infrastructure towards proving \cref{cor:FFTree-existence}.

\begin{proof}[Proof of \cref{cor:Waterhouse}]
 We will ignore at first the condition $|E| > 2\nn$.
 By \cref{thm:Waterhouse}, it is enough to show that there exists an integer $t$ such that
 $\nn \mid q + 1 -t$,  $|t| \le 2\sqrt{q}$, and $t$ is coprime to $q$.

 Since the closed interval $[q - 2\sqrt{q}+1, q + 2\sqrt{q} + 1]$ has length at least $2\nn$, it
 must contain at least two integers $q + 1 - a, q + 1 - (a+\nn)$ which are both divisible by $\nn$.
 Note that at least one of $a, a+\nn$ must be coprime to the characteristic $p$ of $\fq$:
 indeed, if $p \neq 2$, this follows since their difference $\nn = 2^\rounds$ is not divisible by $p$,
 whereas if $p = 2$, then $\nn \mid q + 1 - a$ implies both $a, a+\nn$ are odd and thus coprime to
 $q$---and in fact $a=1$ simply works, yielding a curve of size $q$ (also known as an
 ``anomalous'' curve) and showing $\nnq \ge \frac{q}{4}$.
%  \eli{what's ``anomalous''? Singular? And does it matter to our discussion here?}
 Thus we can always choose at
 least one of $a, a+\nn$ as our candidate for $t$, for which a corresponding curve exists.

 Finally, to assert $|E| > 2\nn$, note that $|E| > q - 2\sqrt{q}$ and $2\nn \le 4\sqrt{q}$, thus for
 all $q \ge 36$ we get
 $$|E| > q - 2\sqrt{q} \ge 6\sqrt{q} - 2\sqrt{q} \ge 2\nn$$
 as claimed. The finitely many cases of $7 \le q < 36$ can be manually checked to verify
 that indeed for each such $q$ there is an elliptic curve $E$ with size exactly $3\nnq$.
\end{proof}

\begin{remark}
 \cref{cor:Waterhouse} is false for $q = 2, 4, 5$: since $2\sqrt{q}$ is not much smaller than $q$
 for these prime powers, for the largest $\nn$ below $2\sqrt{q}$, we have $q +2\sqrt{q}+1 < 3\nn$,
 and therefore no curve has order divisible by $\nn$ and greater than $2\nn$.
 %  \eli{any reason? reference? i.e., this last line seems mystic and would be good to explain why}
\end{remark}

See also \cite{ShSu17} for an overview of practical algorithms for finding such curves.
We note that restricting the size of $\nn$ further,
e.g.\@ $\nn \le \sqrt{q}$ or even $\nn = o(\sqrt q)$, greatly increases the number of possible
curves, and similarly decreases the difficulty of finding one.

Starting from a curve as guaranteed by \cref{cor:Waterhouse}, we now construct a chain of curves and
isogenies with useful properties.

% \dan{Should we include something about Cohen--Lenstra? I don't really understand it. \cite{ShSu17} refers
% to them, of course.}
%
% % Also Cohen--Lenstra Heuristics on abundance ...
%
% \eli{Note to add somewhere an open question as to getting $N$ to be $\Omega(q)$, rather than $\sqrt{q}$. The abundance of smooth numbers should probably give this...}

%\eli{Notation, for dealing with after submission: perhaps curves below should be $E\zr,E\one,\ldots, E\fin$ as all of our objects are denoted?}

% Let $q$ be a prime power and let $\nn = 2^\rounds$ be a power of two greater than $1$ such that there exists an
% integer $N$ divisible by $\nn$ in the range $q+1\pm 2\sqrt{q}$, for which $\frac{N}{\nn}\ne 1,2$
% and $q + 1 - N$ is coprime to $q$.
% In particular, such $N$ exists if $\nn \leq 2\sqrt{q}$ and $q\geq 7$ or if $4\nn\mid q$.
% We will define a certain special subset $L\zr \subseteq \fq$ based on elliptic curves.

\begin{theorem}\label{thm:curve sequence}
For any prime power $q$ and any $1 < \nn = 2^\rounds \le \nnq$, there exist elliptic curves
$E_0, E_1, \ldots, E_{\rounds}$ over $\fq$ in extended Weierstrass form,
%\eli{i think we should say the curves are in extended Weierstrass form, for the projections to make sense}
a subgroup $G_0\subseteq E_0$ of size $\nn$, $2$-isogenies $\phi_i : E_i \to E_{i+1}$
and rational functions $\psi\ii :\Prj \to \Prj$ of degree 2, such that the following diagram is commutative:

\begin{equation}\label{eq:comdiag}
 \begin{tikzcd}
E_0 \arrow{r}{\phi_0} \arrow[swap]{d}{\pi_0} &
E_1 \arrow{r}{\phi_1} \arrow[swap]{d}{\pi_1} &
\cdots \arrow{r}{\phi_{\rounds-1}} &
E_{\rounds} \arrow{d}{\pi_{\rounds}} \\%
\Prj \arrow{r}{\psi\zr}& \Prj \arrow{r}{\psi\one}&
\cdots \arrow{r}{\psi^{(\rounds-1)}}&
\Prj
\end{tikzcd}
\end{equation}
where:
\begin{itemize}
\item $\pi_i$ are the projection maps to the $x$-coordinate of each curve;
\item $\ker(\phi_{i})\subseteq G_{i} \coloneqq \phi_{i-1}\circ\cdots\circ\phi_0(G_0)$ for all $i$; and
\item $G_0$ has a coset $C$ such that $C\ne -C$ (as elements of the quotient group $E_0/G_0$).
%\eli{$-C$ undefined}
%\dan{it's the quotient group operation, I think this doesn't require definition, but did clarify it}
\end{itemize}

\begin{remark}
 The existence of the coset $C$ with $C \ne -C$ will be crucial in the derivation of
 \cref{cor:Lpsi seq}, i.e.\@ in the construction of the $\FFTree$ structure.
\end{remark}

% \eli{explain why $C\neq -C$ is needed}
\end{theorem}

\begin{proof}
	By the definition of $\nnq$, there exists an elliptic curve $E_0$ over $\fq$ with exactly $N$
	points, where $\nn \mid N$ and $N > 2\nn$. Since $E_0$ is abelian, it has a subgroup of any
	order dividing $N$, in particular of order $\nn$. However, since we want to ensure the existence
	of coset $C$ with $C \neq -C$, we may need to choose $G_0$ more carefully.\footnote{As the proof
	shows, this	is in fact only an issue when $\frac{N}{\nn} = 4$; in other cases any choice of $G_0$
	works.} The proof that an appropriate $G_0$ exists is technical and not of particular
	importance, and the interested reader may find it in \cref{sec:existence of G_0}. We note
	that the condition that $N > 2\nn$ is required exactly to ensure the existence of such $G_0$
	and $C$.

% 	\dan{Maybe move the proof of existence of good $G_0$ to the appendix? It obstructs the
% 	beauty of the construction.}
%     \dan{Moved the proof of existence of good $G_0$ to the appendix---what do you think?}
%     \eli{I think moving it there is fine.}

% 	As noted in \cref{sec:group struct}, $E_0$ is of
% 	rank at most 2, and there is an isomorphism
% 	$$\tau : E_0  \leftrightarrow \Z/(m_12^{l_1}\Z) \times \Z/(m_22^{l_2}\Z) $$
% 	where $m_1, m_2$ are odd with $m_1 \mid m_2$,
% 	$l_1 \le l_2$, $m_1 m_2 2^{l_1+l_2} = N$ and $l_1 + l_2 \ge \rounds$. A subgroup $G_0$
% 	of size $\nn$ will necessarily be of the form
% 	$$G_0 =
% 	\tau^{-1}\parens*{(m_12^{l_1-k_1}\Z)/(m_12^{l_1}\Z) \times (m_2 2^{l_2- k_2}\Z)/(m_22^{l_2}\Z)}
% 	\simeq \Z/2^{k_1}\Z \times \Z/2^{k_2}\Z $$
% 	with $k_1\le l_1$, $k_2\le l_2$ and $k_1 + k_2 = \rounds$, and the quotient $E/G_0$ is then
% 	isomorphic to
% 	$$E_0/G_0 \simeq \Z/(m_12^{l_1-k_1}\Z) \times \Z/(m_22^{l_2-k_2}\Z).$$
% 	We wish to ensure that this group contains an element $C$ such that $C \neq -C$, or equivalently,
% 	$2C \neq 0$. This is clearly the case for any choice of $k_1, k_2$, except if
% 	$m_1 = m_2 = 1$ and $l_1 - k_1, l_2 - k_2 \le 1$. But since
% 	$m_1 m_2 2^{l_1 -k_1 + l_2 - k_2} = \frac{N}{\nn} > 2$,
% 	this happens only when $N = 4\nn$ and for the choice $k_1 = l_1 -1$ and $k_2 = l_2 - 1$.
% 	But, by the assumption $\nn > 1$ and by $l_2 \ge l_1$, we find $l_2 \ge 2$, thus we may choose
% 	instead $k_1 = l_1$ and $k_2 = l_2 - 2$, to obtain $E_0/G_0 \simeq \Z/4\Z$, which contains
% 	an element $C$ with $C \neq -C$.

    Having constructed $G_0$, we choose inside it a subgroup of size 2, and use
    \cref{prop:quotient isogeny} to find a new Weierstrass curve $E_1$ and 2-isogeny
    $\phi_0 : E_0 \to E_1$ whose kernel is the subgroup. Thus $G_1 = \phi_0(G_0)$ is a subgroup
    of $E_1$ of order $2^{\rounds-1}$, and we continue iteratively, at step
    $i$ constructing $E_{i+1}$ and $\phi_i$ such that the kernel of $\phi_i$ is a size 2
    subgroup of $G_i$, the image of $G_0$ in $E_i$, which is of size $2^{\rounds-i}$.
%    \eli{suggest adding statement that since $G_i$ is smaller, having a coset $C$ of it the doesn't equal $-C$ is also doable (even more so).}
    The iteration stops at $E_\rounds$, where the image $G_\rounds$ of $G_0$ becomes a singleton.
    % 	is of odd size for the first time. \eli{isn't the odd size necessarily $1$?}

    By \cref{prop:standard form,def:degree of projection equals degree of isogeny}, having written all curves $E_i$ in extended Weierstrass forms,
    we find that there exist rational functions $\psi\ii$, of degrees equal to $\deg\phi_i = 2$,
    which complete the commutative diagram as claimed.
\end{proof}

Focusing on the bottom row of \eqref{eq:comdiag}, we obtain \cref{cor:FFTree-existence} as a direct
corollary of \cref{thm:curve sequence}. The following theorem is an equivalent reformulation of
\cref{cor:FFTree-existence}, directly recalling the definition of the $\FFTree$.

\begin{theorem}\label{cor:Lpsi seq}
Let $q$ be a prime power, and let $\rounds$ be such that $\nn = 2^\rounds \le \nnq$.
There exist subsets $L\zr, L\one, \ldots, L\fin \subseteq \fq$ and
degree $2$ rational functions $\psi\ii(X) = \frac{u\ii(X)}{v\ii(X)} \in \fq(X)$ such that:
\begin{enumerate}
\item $|L\ii| = 2^{\rounds - i}$.
\item $\psi\ii$ is a $2$-to-$1$ map from $L\ii$ onto $L\iip$.
\end{enumerate}
\end{theorem}

\begin{proof}
The case $\nn=1$ is trivial.
If $\nn >1$, apply \cref{thm:curve sequence} to find $E_i, \phi_i, \psi\ii$ and $G_0$ as above,
and let $C$ be a coset of $G_0$ such that $C\ne -C$. For each $i$ define
$C_i$ to be the image of $C$ in $E_i$, i.e.\@ $C_i = \phi_{i-1}\circ\cdots\circ\phi_1\circ\phi_0(C)$.
Since $\ker(\phi_{i-1}\circ\cdots\circ\phi_0) < G_0$, by the third isomorphism theorem,
the map $\phi_{i-1}\circ\cdots\circ\phi_0$
induces an embedding $E_0/G_0 \hookrightarrow E_i/G_i$ which maps distinct cosets of $G_0$ to
distinct cosets of $G_i$, and $C$ to $C_i$.
In particular $C \neq -C$ as cosets of $G_0$ implies to $C_i \neq -C_i$ as cosets of $G_i$.
% \eli{suggest defining $C_i$ inductively, to clarify}
% the kernel of the composed map from $E_0$ to $E_i$ is a subgroup of $G_0$.
% \eli{last deduction too fast for me, why does kernel being subgroup imply that $C_i\neq -C_i$?}
% \dan{is it clear now?} \eli{intuitively, yes, though I wonder if there's a quotable principle that says this. ``This'' is: if the kernel of a homomorphism is strictly contained in a group $G$, then distinct cosets of $G$ are mapped to distinct cosets in the image group}
Define $L\ii=\pi_i(C_i)$. Note that since $C_i, -C_i$ are cosets, $C_i \neq -C_i$ means they are
disjoint, and thus $\pi_i$ is a $1$-to-$1$ map from $C_i$ onto $L\ii$.
In particular $|L\ii|=|C_i|=2^{\rounds-i}$.

Finally, since the diagram is commutative and $\phi_i$ is a $2$-to-$1$ map from $C_i$ onto $C_{i+1}$,
$\psi\ii$ is a $2$-to-$1$ map from $L\ii$ onto $L\iip$.
\end{proof}

% For future use, we note a consequence of the above corollary in the language of $\FFTrees$.
% \begin{corollary}\label{cor:FFTree-existence}
% Let $q$ be a prime power, and let $\rounds$ be such that $\nn = 2^\rounds \le \nnq$; in particular,
% one may take $\nn$ to be any power of two up to $2\sqrt{q}$ for $q \ge 7$.
%
% Then there exists an $\FFTree$ over $\fq$ with depth $\rounds$.
% \end{corollary}

\begin{remark}\label{rem:ffforest}
Not to miss the forest for the trees, we clarify some features of this
elliptic curve based construction.
A careful examination of the proof in \cref{sec:existence of G_0} shows that $C = -C$ holds for
at most $4$ different cosets. % (under certain conditions, the bound $4$ can be replaced with $2$ or $1$).
% \footnote{In fact, if % $\nn$ is sufficiently large, specifically
% $\nn \ge 2^{l_1}$ in the notation of \cref{sec:existence of G_0}, then $G_0$
% can be chosen such that $C=-C$ holds for at most 2 cosets; and if $\nn = 2^{l_1+l_2}$,
% i.e.\@ $\frac{N}{\nn}$ is odd, then $C = -C$ holds only for the trivial coset $C = G_0$.}.
% \david{how much close? and why do we need it to be close?}
% \david{are you sure that 4 is not possible? it fits the $N/2\nn-2$ later}
% \dan{Re David's questions: The $N/2\nn - 2$ is indeed computed according to 4 cosests as it is indeed
% a possible case for small-ish $\nn$. Anyway David's questions made me realize that this unimportant
% comment really belongs in a footnote and not inside the main text.}
% \david{I still don't like the sufficiently large thing, because it makes the reader think that for
% small $\nn$ there are less such cosets, while the contrary is true. I suggest removing the condition
% on $\nn$ and letting the number of problematic cosets be at most 4, which is still very low without
% getting into the different cases of when exactly we get 1, 2 or 4}
The rest of the cosets appear in pairs $\{C^{(j)}, -C^{(j)}\}$, each pair projecting through
$\pi_0$ to a different (and disjoint) $L_j\zr = \pi_0(C\jj) = \pi_0(-C\jj)$.
Thus, our construction actually yields at least $\frac{N}{2\nn} - 2$ different $\FFTrees$,
with pairwise disjoint vertices from all trees at every fixed level, but with the same rational
functions $\psi\ii$ across all trees.

Thus, there exists not only a single $\FFTree$, but an entire $\FFForest$ of disjoint $\FFTrees$
all sharing the same maps. The algorithms in \cref{sec:algorithms} will all be described for the
case of a single $\FFTree$ and subsets of its vertices, but we note that many of them can also
be applied without additional complexity on sets taken from two (or $O(1)$) different $\FFTrees$
belonging to the same $\FFForest$. Note that the total number of leaves in this $\FFForest$ is
$\Omega(q)$, or, more accurately, $\frac{q}{2} - O(\sqrt{q} + \nn)$.
\end{remark}

%\swastik{This is a good place to mention that
%this construction actually constructs an $\nice{FFForest}$, where
%all the trees in the forest look like the above constructed $\FFTree$,
%but the total number of leaves in the forest is as large as $\Omega(q)$.}

%\swastik{Conjecture - if $s$ is a double root of $\psi(X) = t$, then
%by evaluating $P(X)$ and $P'(X)$ at $s$, we can recoverer
%$P_0(t)$ and $P_1(t)$. If so, then maybe we can enlarge $L$, and further avoid a lot of the extra care that we currently take about $C \neq -C$. }
%\dan{You will probably still need to avoid at least the coset $C = G_0$, since $\pi$ has a double
%pole at $O$. Anyway I don't see why it makes $L$ larger?}

\section{Representing polynomials via \FFTrees}
\label{sec:representations}

In this section, we show how to use $\FFTrees$ to
get a nice representation for polynomials that supports
fast operations.

We begin by fixing an $\FFTree$ for the rest of this section.
Thus we have sets $L\zr, L\one, \ldots, L\fin \subseteq \fq$,
and degree-$2$ rational functions $\psi\ii : L\ii \to L\iip$.
We let $L = L\zr$ and let $\nn = |L| = 2^\rounds$.
Also recall the associated binary tree $\calF$ whose set of leaves is $L$.

 All the data structures and algorithms for polynomials
that we describe will be in the context of this $\FFTree$.
While the exact details of how this $\FFTree$ is obtained are not important for anything in this section,
it will be helpful to recall the parameters of $\FFTrees$ that are achievable via \cref{cor:FFTree-existence}.

\subsection{Evaluation tables}
% \dan{Shouldn't this appear under the notations' section?}
% \eli{prefer it here. Explanation: I like ``Notation'' to have things that are well-known, knowing that readers will skip that section and refer to it only as reference when they don't understand something. But if you have notation that's quite different and really essential to what's going on, better put it right before its first use, as we did here.}
% \swastik{I think keeping it here is better, for Eli's reason.}
% \dan{OK}

We shall represent polynomials by their evaluations on various special sets of points, so we
introduce a special notation that will emphasize the sets of evaluation points used. Concretely, an {\em evaluation table} is specified by the following data:
\begin{itemize}
\item a set $S \subseteq \fq$,
\item a function $f: S \to \fq$.
\end{itemize}
We denote the associated evaluation table by $\EVT{f}{S}$, pronounced ``$f$ on $S$''.

For a polynomial or rational function $P(X) \in \fq(X)$ with $P(X)$ defined on $S$,
we define the associated evaluation table
$\EVT{P}{S}$ to be the evaluation table $\EVT{\left.P\right|_S}{S}$, where
$\left.P\right|_S$ is the function from $S$ to $\fq$ given by evaluation of $P$. Looking ahead,
we shall use evaluation tables for operations like
\begin{itemize}
\item Adding, multiplying and dividing, as in this example: given $\EVT{f}{S}$, $\EVT{g}{S}$, $\EVT{h}{S}$, $\EVT{P}{S}$ for some $P(X) \in \fq[X]$, we can compute $\EVT{\frac{f + P(X) g}{h}}{S}$.
\item Restricting an evaluation table $\EVT{f}{S}$ to a subset $S_0 \subseteq S$,
denoting the restricted table by $\EVT{f}{S_0}$
\item Partitioning a set $S$ into $S = S_0 \cup S_1$, and ``splitting''  $\EVT{f}{S}$ into  $\EVT{f_0}{S_0}$ and $\EVT{f_1}{S_1}$, as well as doing the inverse operation of forming the combined evaluation table $\EVT{f}{S} = \EVT{f_0}{S_0} \cup \EVT{f_1}{S_1},$
where $f:S \to \fq$ is given by:
$$\left.f\right|_{S_0} = f_0,$$
$$\left.f\right|_{S_1} = f_1.$$
\end{itemize}

\subsection{Basic sets and moieties}

%For any set $L\subseteq \fq, |L|=\nn$ there is a bijection between the space of functions $\fq^L$ and the space of polynomials of degree less than $\nn$ over $\fq$. \emph{Evaluation} is the transformation mapping a polynomial $P(X)$ to its evaluation on $L$ and \emph{interpolation} is the dual map. We denote by $\EVT{P}{L}$ the representation of a polynomial $P(X)$ by its evaluation on $L$ and by $\INT{f}{L}$ the polynomial $P(X)$ that interpolates $f\in\fq^L$, where $P(X)$ is now represented by its coefficients $P(X)=\sum_{i<\nn} a_i X^i$.

%In what follows we shall fix $L\zr, L\one, \cdots, L^{(\rounds)} \subseteq \fq$ and the maps $\psi\ii(X) \in \fq(X)$ as given by \cref{cor:Lpsi seq} above, and represent polynomials of degree less than $\nn$ by their evaluation on $L = L\zr$.

%Most of our results will be expressed in terms of $L = L\zr$, though the remaining $L\ii$ (and the maps $\psi\ii$) will play an internal role in the recursive algorithms.

We now identify some important subsets of $L$.

\begin{definition}[Basic sets]
	\label{def:basic set}
	We define a {\bf basic} set to be a subset $S$ of $L$
	which is the set of all descendants in $L$ of some vertex of $\calF$.

	Equivalently, it is a set of size $2^a$ for some integer $a$,
	such that if we let $g$ denote the
	composed function $\psi^{(a-1)} \circ
	\psi^{(a-2)} \circ \cdots \circ \psi\one \circ \psi\zr$,
	then $S = g^{-1}(u)$ for some $u \in L^{(a)}$.
\end{definition}

We have the following important property of basic sets: they can be
partitioned into two basic sets of equal size.
\begin{lemma}\label{lem:basic set partition}
	Any basic set $S$ of size $2^a \geq 2$ can be partitioned to two basic sets $S_0 \cup S_1$, where each $S_i$ has size $2^{a-1}$.
	%Moreover,
%	letting
%	$g:=\psi^{(a-1)} \circ
%	\psi^{(a-2)} \circ \cdots \circ \psi\one \circ \psi\zr$ and
	%assuming $S = g^{-1}(u)$ for $g,u$ as defined in \cref{def:basic set}, and letting
	%$\tilde{g}:=\psi^{(a-2)} \circ \cdots \circ \psi\one \circ \psi\zr$ and $\left(\psi^{(a-1)}\right)^{-1}(u)=\{u_0,u_1\}$, we have $S_0=\tilde{g}^{-1}(u_0)$ and $S_1=\tilde{g}^{-1}(u_1)$.
%	Finally, $\psi\zr(S)$ is a basic set of size $|S|/2$.
\end{lemma}
The proof is immediate from \cref{def:fftree,def:basic set}: if $S$ is the set of all descendants in $L$ of the vertex $u \in \calF$,
then letting $\{u_0, u_1\}$ be the children of $u$, we can take $S_i$ to be the set of all descendants in $L$ of $u_i$.
We shall call $S_0$ and $S_1$ the {\bf \emph{moieties}} of $S$. Note that the two moieties are equivalent,
and can be labeled $S_0, S_1$ or $S_1, S_0$ interchangeably.

The following property of sets and polynomials
with respect to moieties of basic sets will also prove to be important in the paper, especially for
algorithms related to modular arithmetic:
\begin{definition}
 Let $S$ be a basic set, and let $A \subset \fq$ be an arbitrary set. We say $A$ is \emph{half-disjoint}
 from $S$ if it is disjoint from at least one moiety of $S$. Similarly, we say a polynomial $P(X)$
 is half-disjoint from $S$ if its set of zeros is disjoint from at least one moiety of $S$.
\end{definition}

%\begin{definition}[Basic sets]\label{def:basic set}
%We define a {\bf basic} set to be a subset $S$ of some layer $L\ii$ of $\calF$ which is the set of all descendants in $L\ii$ of some vertex of $\calF$.

%Equivalently, it is a set of size $2^a$ for some integer $a$,
%such that if we let $g$ denote the
%composed function $\psi^{(i+a-1)} \circ
%\psi^{(i+a-2)} \circ \cdots \circ \psi^{(i+1)} \circ \psi\ii$,
%then $S = g^{-1}(u)$ for some $u \in L^{(i+a)}$.
%Notice that $S'=\psi\ii(S)$ is a basic set of size $2^{a-1}$, we call it the basic set that lies \emph{above} $S$ and we say $\psi\ii$ \emph{generates $S'$ from $S$}.
%\end{definition}

We now consider representations of polynomials by evaluation tables.
Since nonzero polynomials of degree $< n$ cannot vanish in $n$ points, we
immediately get the following fundamental fact.
For distinct polynomials $P(X), Q(X) \in \fq[X]$ with $\deg(P), \deg(Q) < n$,
and a set $S$ with $|S| = n$, we have
that
$$\EVT{P}{S} \neq \EVT{Q}{S}.$$
Thus, for a fixed set $S$ with $|S| = n$,
$\EVT{P}{S}$ is a way of representing a polynomial $P$ with
degree $< n$. The key to our fast algorithms for working with
such a representation is to choose $S$ to be a basic set.

We now define a standard representation for polynomials (in the context of the fixed $\FFTree$). This standard representation will support fast operations, and will be used when we describe applications to classical problems.

For each $a \leq \rounds$, we arbitrarily pick a basic set $U_a$
with size $2^a$ such that:
$$ U_0 \subseteq U_1 \subseteq \cdots \subseteq U_{\rounds} = L.$$
We will call this $U_a$ the {\em standard} basic set of size $2^a$.

%\dan{I don't like the use of  ``canonical'' to mean an arbitrarily-chosen set, this is
%quite opposite from the usual meaning of canonical. Perhaps we can use a different word?
%Some suggestions (in descending order of personal preference):
%standard, default, special, priviliged, marked, ...}
%\swastik{Agreed, made it standard. }
%\eli{not sure we ever use the term standard in this context anymore...}

For a polynomial $P(X)$ and an integer $a$ with $2^a > \deg(P)$,
we define the {\em standard representation of $P$
at scale $a$}, denoted $\canon{P}{a}$, to be $\EVT{P}{U_a}$.

For a polynomial $P(X)$, we define {\em {\bf the} standard representation of $P$},
to be the $\canon{P}{a_0}$, where $a_0$ is the smallest integer
with $2^{a_0} > \deg(P)$.

This standard representation will be our data structure for representing polynomials. In the next section, we show how the $\FFTree$ enables fast operations for this representation
of polynomials.

\section{Fast polynomial algorithms from \FFTrees}
\label{sec:algorithms}

As in the previous section, we assume that we have fixed an $\FFTree$.
Again, the exact details of how this $\FFTree$ is obtained is not important for anything in this section, but it will be helpful to recall the parameters
of $\FFTrees$ that are achievable via~\cref{cor:FFTree-existence}.

In this section we give a number of fast algorithms for working with
polynomials $P(X)$ represented using evaluation tables $\EVT{P}{S}$, where
$S$ is a basic set. Inspection will reveal that nearly all of these algorithms can
be converted to arithmetic circuits over $\fq$ with constant fan-in and size that matches the proclaimed
running time (the only exception is the computation of polynomial degree, which outputs an integer, not a
field element). Thus, henceforth when we say an algorithm ``runs in time $t(n)$'' we shall allow it to receive
advice that will be explicitly stated,
and also mean that it can be computed by an arithmetic circuit over $\fq$ with $t(n)$ gates (and constant fan-in).
In particular, we assume each basic arithmetic operation ($+, -, \times, /$) over $\fq$ has constant computational
cost. While the algorithms of this section use division of elements in $\fq$ for clarity, by inspecting the
details it can be seen that they can be reformulated to avoid division by taking advice in a different form
(for example, taking $\EVT{\frac{1}{f}}{S}$ as advice instead of $\EVT{f}{S}$ as advice).

\paragraph{Algorithmic notations}
We use the notation $\ALG_{P_1, P_2,\ldots}(I_1,I_2,\ldots)$ for our algorithms/circuits. $\ALG$ is the name of
the algorithm, the subscript elements $P_1,P_2,\ldots$ denote fixed parameters that affect
constants of the algorithm/circuit and the inputs $(I_1,I_2,\ldots)$ are given inside the parenthesis, and
are variables.
In particular, any data which depends only on $P_1,P_2,\ldots$ can be assumed to be included
as part of the circuit, or given by a precomputation advice, and our running times exclude the time required
to obtain these parameters and constants. Furthermore, $q$ and the $\FFTree$ that we fixed are always assumed
to be part of the fixed parameters of the algorithm.

\paragraph{Directory of algorithms}

Below we give a list of the algorithms in this section.
\begin{enumerate}
\item $\EXTEND_{S,S'}$ which does low degree extension of polynomial evaluations
from a basic set $S$ to another basic set $S'$. $\EXTEND$ is the basis for all the remaining algorithms
in this section.
\item $\MULT$, which multiplies polynomials in the new representation (allowing for the possibility of
the degree growing). Addition is trivially done in linear time so we do not explicitly describe it.
\item $\MEXTEND$, a version of $\EXTEND$ for monic polynomials of known, fixed degree.
\item $\DEGREE$, which computes the degree of a polynomial given in the new representation.
\item $\REDC$, which performs Montgomery reduction---a technical operation that helps with the remaining operations.
\item $\MODV$, which performs modular reduction, reducing a given polynomial in the new representation modulo a fixed polynomial.
\item $\DIV$, which finds the quotient after division by a fixed polynomial.
\item $\ENTER$ and $\EXIT$, which convert between the new representation and the standard monomial representation.
\item $\CRT$ which computes one direction of the Chinese Remainder Theorem, constructing a polynomial
from its residues modulo two fixed and relatively prime polynomials. (The other direction of the CRT
can be done by $\MODV$.)
\end{enumerate}

\subsection{Low degree extension}\label{sec:LDE}

Our first primitive extends the evaluation of $P$ from one basic set to another basic set of the same size in time $O(n \log n)$ (i.e., via an arithmetic circuit over $\fq$ with constant fan-in and $O(n \log n)$ gates). In other words, the algorithm performs Reed--Solomon encoding in quasi-linear time, as long as the message is provided by the evaluation of $P$ on a basic set, and is encoded by evaluating $P$ on a constant collection of basic sets. Such low-degree extensions are often used to produce interactive proofs and interactive oracle proofs.

\begin{theorem}[Low-degree extension]
	\label{thm:low degree extension}
For any two basic sets $S, S'\subset \fq$ with $|S| = |S'| = n$, there is
an algorithm that runs in time
$O(n \log n)$, denoted $\EXTEND_{S,S'}$, which when given as input:
\begin{itemize}
\item $\EVT{P}{S}$, where $P(X) \in \fq[X]$ with $\deg(P) < n$,
\end{itemize}
outputs $\EVT{P}{S'}$.
\end{theorem}

%This is a special case of the following more general theorem: the
%added generality allows a proof by induction. This is the only place
%where we deal with $i$-basic sets for $i \neq 0$.
%
%
%\begin{theorem}
%For any two $i$-basic sets $S, S'$ with $|S| = |S'| = n$,
%there is an algorithm $\EXTEND_{S,S',i}$, which when given as input:
%\begin{itemize}
%\item $\EVT{P}{S}$, where $P(X) \in \fq[X]$ with $\deg(P) < n$,
%\end{itemize}
%runs in time $O(n \log n)$ and computes $\EVT{P}{S'}$.
%\end{theorem}

For the proof of this theorem (and only for this proof) we need a generalization of basic sets:

\begin{definition}[$i$-basic sets]
	\label{def:i basic set}
	We define an $i$-{\bf basic} set to be a subset $S$ of $L\ii$
	which is the set of all descendants in $L\ii$ of some vertex of $\calF$.

	Equivalently, an $i$-basic set is a subset $S$ of $L\ii$
	of size $2^a$ for some integer $a$, such that if we let $g$ denote the function
	$$\psi^{(a+i-1)} \circ \psi^{(a+i-2)} \circ \cdots \circ \psi^{(i+1)} \circ \psi\ii,$$
	then $S = g^{-1}(u)$ for some $u \in L^{(a+i)}$.
\end{definition}

Notice that $0$-basic sets are simply basic sets per \cref{def:basic set}. In our proof, stated next, we shall use the property that
for every $i$-basic set $S$,
$\psi\ii(S)$ is an $(i+1)$-basic set $T$ of size $|S|/2$, and we say $T$ \emph{lies above} $S$ and is \emph{induced by} $\psi\ii$.

\begin{proof}[Proof of \cref{thm:low degree extension}]

We give a more general algorithm $\EXTEND_{S,S',i}$ to solve
the analogous extension problem where
$S$ and $S'$ are $i$-basic sets with $|S| = |S'| = n$. The algorithm $\EXTEND_{S,S'}$ claimed in \cref{thm:low degree extension} is obtained by fixing $i=0$, i.e., $\EXTEND_{S,S'}(\EVT{\pi}{S})=\EXTEND_{S,S',0}(\EVT{\pi}{S})$.

The $\EXTEND_{S,S',i}$ algorithm uses the map $\psi\ii$ to reduce the extension problem for $i$-basic sets of size $n=2^a$ to two analogous extension problems for $(i+1)$-basic sets of
size $n/2$,
%using the partition of \cref{lem:basic set partition},
and then proceeds recursively, by induction on $a$.

%Concretely, let $$g:=\psi^{(i+a-1)} \circ
%\psi^{(i+a-2)} \circ \cdots \circ \psi^{(i+1)} \circ \psi\ii$$
%as in \cref{def:basic set}, and suppose $S=g^{-1}(u)$ as in \cref{lem:basic set partition}.
Let $T = \psi\ii(S)$, $T' = \psi\ii(S')$ be the $(i+1)$-basic sets above $S,S'$, respectively, which are induced by $\psi\ii$.
By \cref{lem:decomposition}, there are unique polynomials
$P_0(X), P_1(X)$ of degree $< n/2$ with:
\begin{align}
\label{eqPdecomp}
P(X) = \left(P_0(\psi\ii(X)) + X P_1(\psi\ii(X)) \right) (v\ii(X))^{\frac{n}{2}-1}.
\end{align}

$\EXTEND_{S,S',i}$ first computes $\EVT{P_0}{T}$ and $\EVT{P_1}{T}$.
(Since $|T| = n/2$, these uniquely determine $P_0(X)$ and $P_1(X)$).
Then it runs $\EXTEND_{T,T',i+1}$ on this to get
$\EVT{P_0}{T'}$ and $\EVT{P_1}{T'}$, and combines the results
to get $\EVT{P}{S'}$.

The algorithm takes as advice $\EVT{(v\ii(X))^{\frac{n}{2}-1}}{S}$, which can be precomputed since it only depends on $S$ and $\psi\ii$, along with
whatever advice is needed in the recursive calls.
\\ \\

\noindent{\bf Algorithm $\EXTEND_{S,S',i}$:}\\
\noindent{\bf Input:} an evaluation table $\EVT{\pi}{S}$
\begin{enumerate}
\item If $n = 1$ (recall that $n = |S| = |S'|$), then
\begin{enumerate}
\item Let $S = \{s\}$ and $S' = \{s'\}$.
\item Define
$$ \pi': S' \to \fq$$ by
$\pi'(s') = \pi(s)$.
\item Return $\EVT{\pi'}{S'}$.
\end{enumerate}
\item Let $T = \psi\ii(S), T' = \psi\ii(S')$ be the sets that lie above $S$ and $S'$ respectively.
\item For each $t \in T$:
\begin{enumerate}
	\item Define $s_0, s_1$ to be the $\psi\ii$-preimages of $t$ (noticing they are distinct because $S$ is a basic set)
	\item Compute $(\pi_0(t), \pi_1(t))=M_t(\pi(s_0),\pi(s_1)))$ where $M_t$ is defined in \cref{eq:M}.
%by solving the two linear equations:
%	\begin{align}
%	\label{eqpi1}
%		\pi(s_0) = \left(\pi_0(t) + s_0 \pi_1(t) \right) v(s_0)^{n/2-1} \\
%	\label{eqpi2}
%		\pi(s_1) = \left(\pi_0(t) + s_1 \pi_1(t) \right) v(s_1)^{n/2-1}.
%	\end{align}
\end{enumerate}
\item Form the evaluation tables $\EVT{\pi_0}{T}$ and $\EVT{\pi_1}{T}$.
\item Let $\EVT{\pi_0'}{T'}$ and $\EVT{\pi_1'}{T'}$ be the evaluation tables
returned by:
$$ \EXTEND_{T,T',i+1}(\EVT{\pi_0}{T})$$
$$ \EXTEND_{T,T', i+1}(\EVT{\pi_1}{T})$$
\item For each $s' \in S'$, define $\pi'(s')$ by
\begin{align}
\label{eqpiprime}
\pi'(s') = \left( \pi_0'(\psi\ii(s')) + s' \cdot \pi_1'(\psi\ii(s')) \right) v\ii(s')^{\frac{n}{2} - 1}.
\end{align}
\item Return $\EVT{\pi'}{S'}$.
\end{enumerate}

\paragraph{Correctness:}
Suppose $P(X) \in \fq[X]$ is a polynomial of degree $<n$.
We want to show that $\EXTEND_{S,S',i}(\EVT{P}{S})$ returns $\EVT{P}{S'}$.

The main claim is that when the input $\EVT{\pi}{S}$
is $\EVT{P}{S}$, the functions
$\pi_0, \pi_1:T \to \fq$ computed by the algorithm  satisfy:
$$\EVT{\pi_0}{T} = \EVT{P_0}{T},$$
$$\EVT{\pi_1}{T} = \EVT{P_1}{T},$$
where $P_0, P_1$ are as in Equation~\eqref{eqPdecomp}.
This is trivially correct for $a=0$ (i.e., when $n=1$) so we focus henceforth on larger values of $n = 2^a$.

Take any $t$ in $T$, and take $s_0, s_1 \in S$ with $\psi\ii(s_0) = \psi\ii(s_1) = t$.
Using the fact that $\pi(s_0) = P(s_0)$ and $\pi(s_1) = P(s_1)$, and the definition of $M_t$ from \cref{eq:M}, \cref{lem:locality} implies that $\pi_0(t) = P_0(t)$ and $\pi_1(t) = P_1(t)$.
Thus
$$\EVT{\pi_0}{T} = \EVT{P_0}{T},$$
$$\EVT{\pi_1}{T} = \EVT{P_1}{T}.$$

By induction on $a$, we conclude that $\EXTEND_{T,T',i+1}$
on $\EVT{P_0}{T}$ and $\EVT{P_1}{T}$
returns $\EVT{P_0}{T'}$ and $\EVT{P_1}{T'}$.

Thus
$$\EVT{\pi_0'}{T'}  = \EVT{P_0}{T'},$$
$$\EVT{\pi_1'}{T'}  = \EVT{P_1}{T'},$$

Using this along with Equations~\eqref{eqpiprime} and~\eqref{eqPmain}, we get that
$$\EVT{\pi'}{S'} = \EVT{P}{S'},$$
as desired. This completes the proof of correctness.

\paragraph{Running time:}
By inspection, we see that our algorithm uses $O(n)$ arithmetic operations over $\fq$ to reduce an instance of $\EXTEND$ of size $n$ to two instances of size
$n/2$.
(Recall that the algorithm fixes various constants, like $v(s)^{n/2-1}$ and the values of the matrix $M_t$.)
Thus the total running time $F(n)$ of this algorithm
satisfies the recursion:
$$ F(n) \leq 2 F(n/2) + O(n).$$
We conclude the running time (or circuit size) is  $O(n \log n)$ and this completes our proof.
\end{proof}

\begin{remark}
 The $\EXTEND$ algorithm as described is defined for $S, S'$ which are basic sets of the same
 size in the same $\FFTree$. However, we note that it works just as well when $S, S'$ are
 basic sets of the same size from two different $\FFTrees$ in the same $\FFForest$ (see also \cref{rem:ffforest}).
\end{remark}

% \dan{Add $\MEXTEND$ Here.}

\subsection{Multiplication}

We give a quick application of the previous algorithm to multiplication of polynomials in the new representation.

\begin{theorem}[Multiplication]
Let $S$ be a basic set with $|S| = n$.
Let $S_0 \subseteq S$ be a moiety of $S$. %a basic set with $|S_0| = n/2$.

There is an algorithm $\MULT_{S, S_0}$, which when given as input:
\begin{itemize}
\item $\EVT{P}{S_0}$, where $P(X) \in \fq[X]$ with $\deg(P)<n/2$, and
\item $\EVT{Q}{S_0}$, where $Q(X) \in \fq[X]$ with $\deg(Q)<n/2$,
\end{itemize}
runs in time
$$O(n \log n)$$
and computes $ \EVT{P\cdot Q}{S}$.
\end{theorem}
\begin{proof}
The algorithm is basically immediate given $\EXTEND$.
Let $S_1$ be the other moiety of $S$.
We first run $\EXTEND_{S_0, S_1}$ on $\EVT{P}{S_0}$ and $\EVT{Q}{S_0}$
to get $\EVT{P}{S_1}$ and $\EVT{Q}{S_1}$. Combining these, we get
$\EVT{P}{S}$ and $\EVT{Q}{S}$, and by pointwise multiplication we get
$\EVT{P\cdot Q}{S}$. The running time comes from two invocations of $\EXTEND$
and $O(n)$ other operations, and is thus $O( n \log n)$.
\end{proof}

\subsection{Monic polynomial extension}\label{sec:MEXTEND}
As noted before, for a set $S$ of size $n$, the linear space of all possible evaluation tables $\EVT{P}{S}$
is in one-to-one correspondence with the space of all polynomials $P(X)$ of degree $< n$. It is also
interesting to note that these spaces are in one-to-one correspondence with the set of all \emph{monic}
polynomials of degree \emph{exactly} $n$. In fact, if $Z(X)$ is the vanishing polynomial of $S$,
and $P(X), Q(X)$ are polynomials with $\deg(Q) < n =\deg(P)$ and $P$ is monic, then
$\EVT{P}{S} = \EVT{Q}{S}$ if and only if $P(X) = Q(X) + Z(X)$.

This property allows us to easily adapt the $\EXTEND$ algorithm into an extension algorithm for
monic polynomials, which we call $\MEXTEND$.
\begin{theorem}[Monic polynomial extension]
	\label{thm:mextend}
For any two basic sets $S, S'\subset \fq$ with $|S| = |S'| = n$, there is
an algorithm that runs in time
$O(n \log n)$, denoted $\MEXTEND_{S,S'}$, which when given as input:
\begin{itemize}
\item $\EVT{P}{S}$, where $P(X) \in \fq[X]$ is monic with $\deg(P) = n$,
\end{itemize}
outputs $\EVT{P}{S'}$.
\end{theorem}

\begin{proof}
 Let $Z(X)$ be the vanishing polynomial of $S$.
 As noted above, for such polynomials $P(X)$, we have $\EVT{P}{S} = \EVT{P-Z}{S}$, and
 $\deg(P(X)-Z(X)) < n$. By the properties of $\EXTEND$ it thus follows that
 $$\EXTEND_{S,S'}(\EVT{P}{S}) = \EXTEND_{S,S'}(\EVT{P-Z}{S}) = \EVT{P-Z}{S'} $$
 and adding $\EVT{Z}{S'}$ pointwise yields
 $$\MEXTEND_{S, S'}(\EVT{P}{S}) \coloneqq \EXTEND_{S, S'}(\EVT{P}{S}) + \EVT{Z}{S'} = \EVT{P}{S'}$$
 as needed. The algorithm takes $\EVT{Z}{S'}$ as advice, calls $\EXTEND$ once and does an additional
 $O(n)$ operations, thus runs in time $O(n \log n)$.
\end{proof}

This algorithm can replace $\EXTEND$ in applications where the polynomials are known to be monic
and of known degrees, with more efficient run times.
For example, it can be used to multiply two monic polynomials of degree $n/2$, represented as
evaluation tables on a set of size $n/2$, with the product similarly being a monic polynomial of degree $n$,
represented as an evaluation table on a set of size $n$. If we were instead to multiply such polynomials
using the standard $\EXTEND$ algorithm, we would have to represent each polynomial by its values on a
set of size $n$, and their product on a set of size $2n$, and use extensions from $n$ to $2n$
instead of extensions from $n/2$ to $n$, which would more than double the required run-time.

\subsection{Degree Computation}
\label{sec:degree computation}

The next operation we describe is that of computing the degree of a polynomial $P$ represented by
its evaluation on a basic set.

\begin{theorem}[Degree Computation]\label{thm:degree computation}
Let $S$ be a basic set of size $|S| = n$.
There is an algorithm $\DEGREE_S$, which when given as input:
\begin{itemize}
\item $\EVT{P}{S}$, where $P(X) \in \fq[X]$ with $\deg(P) < n$,
\end{itemize}
runs in time
$$O(n \log n)$$ and
computes $\deg(P)$.
\end{theorem}
\begin{proof}
%Write $S = S_0 \cup S_1$, where $S_0$ and $S_1$ are basic sets with $|S_0| = |S_1| = |S|/2$. Let
Let $S_0, S_1$ be the moieties of $S$, and let
$Z_0(X)$ be the vanishing polynomial of $S_0$.
The algorithm we give will assume that $\EVT{Z_0}{S_1}$ is given as advice: this is a fixed
precomputation that depends only on $S$.\\
\\

\noindent{\bf Algorithm $\DEGREE_S$:}\\
\noindent{\bf Input:} An evaluation table $\EVT{\pi}{S}$
\begin{enumerate}
\item If $|S| = 1$ with $S = \{s\}$, then
\begin{itemize}
\item if $\pi(s) \neq 0$ return $0$; else, return $-\infty$.
%\eli{I think the correct base case is to always return $0$. The degree of any constant is $0$. Why would it be $-\infty$? After all, for any function we get we know degree is well defined and less than $n$.}
%\dan{The degree of the 0 polynomial is indeed usually defined to be $-\infty$, not 0. There are many
%good reasons why this is the case.}
\end{itemize}
\item Let $\EVT{g}{S_1} = \EXTEND_{S_0, S_1}\left( \EVT{\pi}{S_0} \right)$.
\item \label{caselowdeg} If $\EVT{g}{S_1} = \EVT{\pi}{S_1}$, then return $\DEGREE_{S_1}(  \EVT{\pi}{S_1} )$.

\item \label{casehighdeg} Otherwise, using $ \EVT{\pi}{S_1}$, $ \EVT{g}{S_1}$ and
$\EVT{Z_0}{S_1}$, compute:
$$ \EVT{\frac{\pi - g}{Z_0}}{S_1},$$
and return
$$\frac{n}{2} + \DEGREE_{S_1}\left(\EVT{\frac{\pi-g}{Z_0}}{S_1} \right).$$
\end{enumerate}

\paragraph{Correctness:}
The case $n =1$ is trivial, and when $P$ is the zero polynomial notice by inspection the result will be $-\infty$, as required.

Suppose $n > 1$.
Let $P(X)$ be a polynomial with $0\leq \deg(P) < n$.
Let us consider the execution of the above algorithm on input $\EVT{P}{S}$.

\begin{itemize}
\item {\bf Case 1:} $\deg(P) < n/2$. Then by the defining property of $\EXTEND_{S_0, S_1}$,
we have that
$$\EXTEND_{S_0, S_1}\left( \EVT{P}{S_0} \right) = \EVT{P}{S_1}.$$
Thus in the execution of the algorithm, we will have $\EVT{g}{S_1} = \EVT{P}{S_1}$,
and thus in Step~\ref{caselowdeg} the algorithm will return $$\DEGREE_{S_1}\left( \EVT{P}{S_1} \right),$$
which equals $\deg(P)$ by induction, as desired.

\item {\bf Case 2:} $\deg(P) \geq n/2$. Let $P(X) = R(X) + Z_0(X) \cdot Q(X)$, where $\deg(R) < n/2$.
Thus $\deg(P) = n/2 + \deg(Q)$.

By the above relation between $P$ and $R$, we have
$$\EVT{R}{S_0} = \EVT{P}{S_0} = \EVT{\pi}{S_0}.$$

By the defining property of $\EXTEND_{S_0, S_1}$, we get that $\EVT{g}{S_1} = \EVT{R}{S_1}$.
Thus $$ \EVT{\frac{\pi-g}{Z_0}}{S_1} = \EVT{ \frac{P-R}{Z_0}}{S_1} = \EVT{Q}{S_1},$$
which implies, by induction, that Step~\ref{casehighdeg} returns
$$ n/2 + \DEGREE_{S_1}\left(\EVT{Q}{S_1}\right) = n/2 + \deg(Q) = \deg(P),$$
as desired.
\end{itemize}

\paragraph{Running time:}

The algorithm calls one instance of $\EXTEND$ on an instance of size $O(n)$,
does $O(n)$ operations,
and makes one recursive call to itself on an instance of size $n/2$.
Thus the running time $F(n)$ satisfies:
$$F(n) \leq O(n \log n) + F(n/2),$$
and thus $F(n) \leq O(n \log n)$, as claimed.
\end{proof}

%\dan{c) The algorithm uses both $\EVT{A}{S_0}, \EVT{A}{S_1}$ -- they should be mentioned
%as being inputs. Note that they are not directly extensions of each other, since $A$
%is of degree $=n$. But if we know its leading coefficient, we can correct this; assuming that it is monic (which I believe is true for every usecase?), then
%\begin{align*}
%&\EXTEND_{S_0, S_1}(\EVT{A}{S_0}) = \EXTEND_{S_0, S_1}(\EVT{A \rem Z_0}{S_0}) =
%\EVT{A \rem Z_0}{S_1} = \EVT{A - Z_0}{S_1} \\
%\Rightarrow &\EVT{A}{S_1} = \EXTEND_{S_0, S_1}(\EVT{A}{S_0}) + \EVT{Z_0}{S_1}
%\end{align*}
%}

\subsection{Modular and Montgomery Reduction}\label{sec:mod and montgomery reduction}
\subsubsection{Modular Reduction---theorem statement}
The goal of this chapter is to present an algorithm that computes the remainder of the division
of an input polynomial $P$ (in the new representation) by a fixed polynomial $A$:

\begin{theorem}[Modular Reduction]
	\label{thm:mod}
% Let $S$ be a basic set of size $n$, partitioned as $S = S_0 \cup S_1$.
% Let $A(X) \in \fq[X]$ be a polynomial of degree at most $n/2$,
% having no zeroes in at least one of $S_0$ or $S_1$.
Let $S$ be a basic set of size $n$, and let $A(X) \in \fq[X]$ be a polynomial of degree at most
$n/2$ which is half-disjoint from $S$, i.e.\@ $A(X)$ has no zeroes in at least one moiety of $S$.

There is an algorithm running in time $O(n \log n)$, denoted $\MODV_{S, A}$, which when given as input:
\begin{itemize}
\item $\EVT{P}{S}$, where $P(X) \in \fq[X]$ with $\deg(P) < n$,
\end{itemize}
computes $ \EVT{Q}{S}$, where $Q(X) \in \fq[X]$ is given by:
$$ Q(X) = P(X)  \rem A(X).$$
\end{theorem}
Before presenting the proof and the algorithm, we introduce an auxiliary algorithm, which we call
\emph{Montgomery reduction}, inspired by Montgomery's~\cite{Mont85} algorithm for ``modulo-free''
modular multiplication, which we also describe briefly.

\subsubsection{Montgomery Reduction}\label{sec:montgomery reduction}
Montgomery's algorithm for multiplication is motivated by the observation that while the operation
$a \mmod N$ for a generic (odd) integer $N$ might be computationally expensive,
the operation $a \mmod R$ where $R = 2^r \gtrsim N$ is very efficient, in computing systems based on
binary representations.

In Montgomery's method, each residue $x \pmod N$ is represented instead by $xR \mmod N$.
To get the representation of the product $xy$, i.e.\@ $xyR \mmod N$, we first multiply the two
representations to get an integer equivalent to $xyR^2 \pmod N$, and then apply the \emph{reduction}
algorithm $\REDC$, which efficiently maps an integer $t$ to $tR^{-1} \mmod N$, without explicitly computing
the division by $N$. The reduction algorithm relies on having the constant number $(-N^{-1}) \mmod R$ as
advice.

The representation $xR \mmod N$ can be transformed back to $x \mmod N$ by simply applying reduction.
In the other direction, $x \mmod N$ can be transformed into $xR \mmod N$ by performing the
full Montgomery multiplication (i.e.\@ integer multiplication + reduction) between $x \mmod N$
and the constant $R^2 \mmod N$, which is again given as advice.

For our purposes, we want to perform modular arithmetic of polynomials. We observe that the
vanishing polynomial $Z(X)$ of a basic set $S$ is a natural analogue to the radix $R=2^r$, as
arithmetic operations on the tables $\EVT{P}{S}$ are equivalent to arithmetic operations
on polynomials modulo $Z(X)$. Thus, we can attempt to create a version of $\REDC$ which
transforms $\EVT{P}{S}$ into $\EVT{P\cdot Z^{-1}}{S}$, and then use this algorithm to perform
general modular operations, such as $\MODV$. In fact, we apply $\REDC$ directly only inside
$\MODV$.
%\dan{Maybe comment that both have the same asymptotic run time, but $\REDC$ might have
%independent interest since it has strictly better constants, due to $\MODV$ containing
%two applications of $\REDC$? Probably not interesting for FOCS?}
%\eli{Dan, please go ahead and add such a remark}
%\dan{Added at end of the entire mod section, i.e.\@ before division section}

\begin{theorem}[Montgomery Reduction]\label{thm:invmod}
	Let $S$ be a basic set with $|S|=n$.  Let $S_0\subseteq S$ be a moiety of $S$.
	Let $A(X) \in \fq[X]$ be a polynomial of degree at most $n/2$
	having no zeroes in $S_0$. Let $Z_0(X)$ be the vanishing polynomial
	of $S_0$.

	There is an algorithm running in time $O(n \log n)$, denoted $\REDC_{S, S_0,A}$,
	which when given as input:
\begin{itemize}
\item $\EVT{P}{S}$, where $P(X) \in \fq[X]$ satisfies $\deg(P) < n$,
\end{itemize}
computes $\EVT{Q}{S}$, where $Q(X) \in \fq[X]$ is a polynomial such that
\begin{itemize}
 \item  $Q(X) \equiv P(X) \cdot Z_0(X)^{-1} \pmod {A(X)}$, and
 \item $\deg(Q) \le \max(\deg(P) - n/2,\ \deg(A) - 1) < n/2$.
\end{itemize}
\end{theorem}

\begin{remark}\label{rem:invmod ext}
 If $\deg(P) < n/2 + \deg(A)$, then it follows that $\deg(Q) < \deg(A)$, and therefore
 $$Q(X) = P(X) \cdot (Z_0(X))^{-1}_{A(X)} \rem A(X).$$
 However, the last identity is not true in general when $n/2 + \deg(A) \le \deg(P) < n$, since
 $Q$ might not be of degree less than $\deg(A)$.
\end{remark}

Before giving the algorithm, we give some high level motivation for it.
Suppose for simplicity that $A(X)$ has degree exactly $n/2$.
Observe that $A(X)$ and $Z_0(X)$ are relatively prime.
Thus there exist $G(X), H(X)$ of degree $< n/2$ such that:
$$ G(X)A(X) + H(X) Z_0(X) = P(X).$$
From this identity, if we are given the evaluation of one of $G(X)$ at a point $x \in S\setminus S_0$,
we can compute the evaluation of $H(X)$ at that point $x$.
Finally, we observe that:
\begin{align}
\label{GH1}
G(X) = (P(X) / A(X) ) \rem Z_0(X),\\
\label{GH2}
H(X) =    ( P(X)/Z_0(X) ) \rem A(X).
\end{align}
The first equation tells us how to compute $G(X)$ at any $x \in S_0$, which we can then extend to 
compute evaluations of $G(X)$ at $x \in S \setminus S_0$.
The second equation tells us that $H(X)$ equals $Q(X)$, the polynomial whose evaluations we seek.
Putting these together, we get the algorithm.

\begin{proof}
Let $S_1$ be the other moiety of $S$.
The algorithm uses the values of $\EVT{Z_0}{S_1}$, $\EVT{A}{S_0}$, $\EVT{A}{S_1}$,
which depend only on $S$, $S_0$ and $A$.\\
\\

\noindent{\bf Algorithm $\REDC_{S, S_0, A}$:}\\
\noindent{\bf Input:} an evaluation table $\EVT{\pi}{S}$
\begin{enumerate}
\item From $\EVT{\pi}{S_0}$ and $\EVT{A}{S_0}$, compute
$\EVT{\frac{\pi}{A}}{S_0}.$
\item \label{invmodextstep} Let $\EVT{g}{S_1} = \EXTEND_{S_0, S_1}( \EVT{\frac{\pi}{A}}{S_0} )$.
\item From $\EVT{\pi}{ S_1 }$, $\EVT{ g}{S_1} $, $\EVT{A}{S_1}$, $\EVT{Z_0}{S_1}$,
compute:
$$ \EVT{h_1}{S_1} = \EVT{\frac{\pi - g A}{Z_0}}{S_1}.$$
\item \label{redclastext} Compute:
$$ \EVT{h_0}{S_0} = \EXTEND_{S_1, S_0}(\EVT{h_1}{S_1}).$$
\item Return $\EVT{h_0}{S_0} \cup \EVT{h_1}{S_1}$.
\end{enumerate}

\paragraph{Proof of correctness:}
Let $P(X)$ be a polynomial of degree $ < n$.
We will analyze the above algorithm when its input $\EVT{\pi}{S}$ is taken to be $\EVT{P}{S}$.
Let $G(X) \in \fq[X]$ be the unique polynomial of degree $ < n/2$ interpolating $\frac{\pi}{A}$ on $S_0$; namely:
$$\EVT{G}{S_0} = \EVT{\frac{\pi}{A}}{S_0}.$$

Then by the defining property of $\EXTEND$, we get that:
$$\EXTEND_{S_0, S_1}( \EVT{\frac{\pi}{A}}{S_0} ) = \EXTEND_{S_0, S_1}( \EVT{G}{S_0} ) = \EVT{G}{S_1}.$$
Thus in Step~\ref{invmodextstep} of the algorithm, we will have
$$\EVT{g}{S_1} = \EVT{G}{S_1}.$$

By definition of $G(X)$, we have that  $\frac{P(X)}{A(X)}- G(X)$ vanishes on $S_0$.
Therefore $P(X) - G(X)A(X) \in \fq[X]$ vanishes on $S_0$, and so $Z_0(X)$ divides $P(X) - G(X)A(X)$.
Let $H(X) \in \fq[X]$ be given by:
$$H(X) = \frac{P(X) - G(X)A(X)}{Z_0(X)}.$$
Note that
\begin{equation}\label{eq:H deg}
	\deg(H) \le \max\{\deg(P), \deg(A)+\deg(G)\}-\deg(Z_0) \le \max(\deg(P) - n/2, \deg(A) - 1) < n/2.
\end{equation}
The second inequality follows from the fact that $\deg(G)<\deg(Z_0)=n/2$, and the final inequality from
the assumptions $\deg(P) < n$ and $\deg(A) \le n/2$.
%and the facts $\deg(A)\leq n/2$ and $\deg(G)<n/2$.

We have
$$ \EVT{h_1}{S_1} = \EVT{\frac{ \pi - gA}{Z_0}}{S_1} = \EVT{\frac{P - GA}{Z_0}}{S_1} = \EVT{H}{S_1}.$$

Thus in Step~\ref{redclastext} $\EXTEND_{S_1, S_0}$ yields
$$ \EVT{h_0}{S_0} = \EVT{H}{S_0},$$
and so the algorithm returns:
$$\EVT{h_0}{S_0} \cup \EVT{h_1}{S_1} = \EVT{H}{S_0} \cup \EVT{H}{S_1} = \EVT{H}{S}$$
and we have already shown in \cref{eq:H deg} that $H(X) = Q(X)$ is of the claimed degree.

Finally, from the definition of $H(X)$ we get
$$H(X) Z_0(X) = P(X) - G(X)A(X) \equiv P(X) \pmod{A(X)},$$
which after dividing by $Z_0(X)$ is equivalent to
$$H(X) \equiv P(X)\cdot Z_0(X)^{-1} \pmod{A(X)}.$$

% It remains to show that $H(X) \equiv P(X) \cdot Z_0(X)^{-1}\pmod{A(X)}$.
% This can be deduced from the following two claims:
% \begin{enumerate}
% \item $\deg(H) < \deg(A) = n/2$, which is given by \cref{eq:H deg}.
% \item $H(X) \equiv P(X) \cdot (Z_0(X))^{-1}_{A(X)} \pmod{A(X)}$. This is a consequence
% of the definition of $H(X)$, which implies that
% $$H(X) Z_0(X) = P(X) - G(X)A(X) \equiv P(X) \pmod{A(X)}.$$
% \end{enumerate}

This completes the proof of correctness.
\paragraph{Running time:}
The algorithm does $O(n)$ operations and invokes $\EXTEND$ twice on instances of size $n/2$. Thus the total running time is $O(n \log n)$.
\end{proof}

\subsubsection{Modular Reduction---algorithm and proof}\label{sec:modular reduction}

% The next operation computes the remainder of a polynomial $P$ from its division by another fixed polynomial $A$.
%
% \begin{theorem}[Modular Reduction]
% 	\label{thm:mod}
% Let $S$ be a basic set of size $n$.
% Let $A(X) \in \fq[X]$ be a polynomial of degree at most $n/2$,
% having no zeroes in $S$.
%
% There is an algorithm running in time $O(n \log n)$, denoted $\MODV_{S, A}$, which when given as input:
% \begin{itemize}
% \item $\EVT{P}{S}$, where $P(X) \in \fq[X]$ with $\deg(P) < n$,
% \end{itemize}
% computes $ \EVT{Q}{S}$, where $Q(X) \in \fq[X]$ is given by:
% $$ Q(X) =     P(X)  \rem A(X).$$
% \end{theorem}

\begin{proof}[Proof of \cref{thm:mod}]
Let $S_0, S_1$ be the moieties of $S$, and suppose without loss of generality that $A(X)$
has no zeros in $S_0$ (otherwise, it has no
zeros in $S_1$ by assumption, and we may swap the labeling of the moieties).

Let $C(X) =  Z_0(X)^2 \rem A(X)$, which has degree $\deg(C) < \deg(A)$.
The algorithm uses the values of $\EVT{C}{S}$, that depend only on $A(X)$, $S$ and $S_0$, as
well as values used internally by the $\REDC_{S, S_0, A}$ sub-circuit.
\\
\\

\noindent{\bf Algorithm $\MODV_{S,A}$:}\\
\noindent{\bf Input:} an evaluation table $\EVT{\pi}{S}$
\begin{enumerate}
\item \label{invmod1} From $\EVT{\pi}{S}$, compute
$$\EVT{h}{S} = \REDC_{S, S_0, A}( \EVT{\pi}{S}).$$
\item \label{hR} From $\EVT{h}{S}$ and $\EVT{C}{S}$, compute
$$ \EVT{h \cdot C}{S}.$$
\item \label{invmod2} Compute:
$$\EVT{g}{S} = \REDC_{S, S_0, A}(\EVT{h\cdot C}{S} ).$$
\item Return $\EVT{g}{S}$.
\end{enumerate}

\paragraph{Proof of correctness:}
Suppose $P(X) \in \fq[x]$ with $\deg(P) < n$.
We will analyze the above computation when its input $\EVT{\pi}{S}$ is taken to be
$\EVT{P}{S}$.
By \cref{thm:invmod} about $\REDC$, we get that Step~\ref{invmod1} computes $\EVT{h}{S} = \EVT{H}{S}$,
where $H(X)$ is a polynomial with $\deg H(X) < n/2$ and satisfying
$$H(X) \equiv P(X) \cdot Z_0(X)^{-1} \pmod{A(X)}.$$

Thus Step~\ref{hR} computes
$$\EVT{H\cdot C}{S},$$
where $H(X)\cdot C(X)$ is a polynomial with
$$\deg(H\cdot C) = \deg(H) + \deg(C) < n/2 + \deg(A)$$
and
$$H(X) \cdot C(X) \cdot Z_0(X)^{-1} \equiv P(X) \cdot Z_0(X)^{-1} \cdot Z_0(X)^2 \cdot Z_0(X)^{-1}
\equiv P(X)\pmod{A(X)}.$$

Thus in Step~\ref{invmod2}, as noted in \cref{rem:invmod ext}, the algorithm returns $\EVT{Q}{S}$,
where
$$Q(X) = \left(H\cdot C \cdot (Z_0(X))^{-1}_{A(X)} \right) \rem A(X) = P(X) \rem A(X),$$
as desired.

\paragraph{Running time:}
The algorithm invokes $\REDC$ twice on instances of size $n$, and does $O(n)$ other operations.
Thus the running time is $O(n \log n)$.
\end{proof}

\begin{remark}
 As noted earlier in this section, we have no further direct applications of $\REDC$ in this paper,
 and all calls to it are mediated by calls to $\MODV$. Nonetheless, we note that it may hold individual
 interest for real-world applications, as it is naturally more than twice as fast as $\MODV$, due to
 $\MODV$ containing two calls to $\REDC$. Thus, applying $\REDC$ directly might be more efficient in
 certain situations.
\end{remark}

\subsection{Division}

We give a quick application of the previous algorithm to finding the quotient of an input polynomial
$P$ (in the new representation) by a fixed polynomial $A$.
% \david{something in the grammar of this sentence does not make sense to me}
% \eli{lgtm, as written now}

\begin{theorem}[Division]
Let $S$ be a basic set with $|S| = n$.

Let $A(X)$ be a polynomial with degree at most $n/2$ having no zeroes in $S$.

There is an algorithm $\DIV_{S}$, which when given as input:
\begin{itemize}
\item $\EVT{P}{S}$, where $P(X) \in \fq[X]$ with $\deg(P)<n$, and
\end{itemize}
runs in time
$$O(n \log n)$$
and computes $ \EVT{Q}{S}$,
where $Q(X)$ is the quotient when $P(X)$ is divided by $A(X)$.
\end{theorem}
\begin{proof}
The algorithm is basically immediate given $\MODV$.
Letting $R = \MODV_{S,A}(\EVT{P}{S})$,
the algorithm returns $\EVT{\frac{P-R}{A}}{S}$.
\end{proof}

\subsection{Exiting to Standard Polynomial  Representation}\label{sec:exit}

The next computation transforms  a polynomial represented by its evaluation on a basic set to the set of coefficients that
form the standard representation as $\sum_i a_i X^i$.

\begin{theorem}[Exit to Standard Polynomial Representation]\label{thm:exit}
Let $S$ be a basic set with $|S| = n$.

There is an algorithm $\EXIT_{S}$, which when given as input:
\begin{itemize}
\item $\EVT{P}{S}$, where $P(X) \in \fq[X]$ with $\deg(P) < n$,
\end{itemize}
runs in time
$$O(n \log^2 n)$$
and computes the coefficients $a_i$ of $P(X)$ in the standard expansion
$$P(X) = \sum_{i=0}^{n-1} a_i X^i.$$
\end{theorem}
\begin{proof}
%Write $S = S_0 \cup S_1$ where $S_0, S_1$ are basic sets with $|S_0| = |S_1| = n/2$.
Let $S_0, S_1$ be the moieties of $S$, and note that we may assume without loss of generality that
$0 \notin S_0$.
Thus $X^{n/2}$ has no roots in $S_0$, and in particular an algorithm $\MODV_{S, X^{n/2}}$ exists.

On input $\EVT{P}{S}$, the algorithm will compute $\EVT{U}{S_0}$ and $\EVT{V}{S_0}$,
where $P(X) = U(X) + X^{n/2} \cdot V(X)$ with $\deg(U), \deg(V) < n/2$,
in time $O(n \log n)$. Then by recursively calling $\EXIT_{S_0}$ on these two smaller instances
and combining the results in the obvious way, we get the coefficients of $P(X)$ in time
$O(n \log^2 n )$.

The algorithm uses as advice the values $\EVT{X^{n/2}}{S_0}$, which depend only on $S$, as well as auxiliary values used by the $\MODV$ algorithm (namely, $\EVT{Z_0}{S_1}$, $\EVT{Z_0^2 \rem X^{n/2}}{S}$).
\\
\\

\noindent{\bf Algorithm $\EXIT_{S}$:}\\
\noindent{\bf Input:} an evaluation table $\EVT{\pi}{S}$
\begin{enumerate}
\item If $|S| = 1$ with $S = \{s\}$, return $(\pi(s) )$.
\item \label{exitmodstep} Let $\EVT{u}{S} = \MODV_{S,X^{n/2}}(\EVT{\pi}{S})$.
\item Let
$$(a_0, a_1, \ldots, a_{\frac{n}{2}-1}) = \EXIT_{S_0}(\EVT{u}{S_0}).$$

\item From $\EVT{\pi}{S_0}$, $\EVT{u}{S_0}$ and $\EVT{X^{n/2}}{S_0}$, compute:
$$\EVT{v}{S_0} = \EVT{\frac{\pi - u}{X^{n/2}}}{S_0}.$$

\item  Let
$$(b_0, b_1, \ldots, b_{\frac{n}{2}-1}) = \EXIT_{S_0}(\EVT{v}{S_0}).$$

\item Return
$$(a_0, a_1, \ldots, a_{\frac{n}{2} - 1}, b_0, b_1, \ldots, b_{\frac{n}{2}-1}).$$

\end{enumerate}

\paragraph{Correctness:}
Suppose $P(X) \in \fq[X]$ with $\deg(P) < n$.
We will analyze what the above algorithm does when its input $\EVT{\pi}{S}$ is taken
to be $\EVT{P}{S}$.

If $n = 1$ the algorithm is clearly correct.

Now assume $n  > 1$.
Write $P(X) = U(X) + X^{n/2} \cdot V(X)$, where $\deg(U), \deg(V) < n/2$.

Then $U(X) = P(X) \rem X^{n/2}$.
By properties of $\MODV$, we get that Step~\ref{exitmodstep} computes $\EVT{u}{S} = \EVT{U}{S}$. Thus $\EVT{u}{S_0} = \EVT{U}{S_0}$.

Also note that $V(X) = \frac{P(X) - U(X)}{X^{n/2}}$.
Then
$$ \EVT{v}{S_0} = \EVT{\frac{\pi - u}{X^{n/2}}}{S_0} =  \EVT{\frac{P - U}{X^{n/2}}}{S_0} = \EVT{V}{S_0}.$$

By induction, we get that the algorithm correctly computes the coefficients of $U(X)$ and $V(X)$, and by
concatenating them together, it computes the coefficients of $P(X)$, as desired.

\paragraph{Running time:}

The algorithm makes one call to $\MODV$ on an instance of size $n$ and two recursive calls to $\EXIT$ on instances of size $n/2$.
Thus the running time $F(n)$ satisfies the recurrence:
$$F(n) \leq 2 F(n/2) + O(n \log n),$$
and thus
$$F(n) \leq O( n \log^2 n).$$
\end{proof}

\subsection{Entering from Standard Polynomial Representation}
\label{sec:enter}

The next algorithm is the inverse of $\EXIT$, it transforms a polynomial given in standard representation to its evaluation
over a basic set.

\begin{theorem}[Entering from Standard Polynomial Representation]
	\label{thm:enter}
Let $S$ be a basic set with $|S| = n$.

There is an algorithm $\ENTER_{S}$, which when given as input:
\begin{itemize}
\item $a_0, a_1, \ldots, a_{n-1} \in \fq$,
\end{itemize}
runs in time
$$ O(n \log^2 n)$$
and computes $\EVT{P}{S}$, where
$$P(X) = \sum_{i=0}^{n-1} a_i X^i.$$
\end{theorem}
\begin{proof}
If $|S| = 1$, the task is trivial.

% Otherwise, write $S = S_0 \cup S_1$ where $S_0, S_1$ are basic sets with $|S_0| = |S_1| = n/2$.
Otherwise, let $S_0, S_1$ be the moieties of $S$.
The algorithm is based on writing the polynomial $P(X)$
as:
$$ P(X) = U(X) + X^{n/2}\cdot V(X),$$
where $\deg(U), \deg(V) < n/2$, and finding the evaluation tables of $U, V$ on both $S_0, S_1$.
A priori, this seems like reducing an $\ENTER$ instance of size $n$ to $4$ $\ENTER$ instances of size $n/2$ (leading to a quadratic running time), but in fact this can be done by $2$ recursive calls to $\ENTER$ and $2$ invocations of $\EXTEND$.

The algorithm below takes $\EVT{X^{n/2}}{S}$ as advice. This can be precomputed since it only depends on $S$.\\
\\

\noindent{\bf Algorithm $\ENTER_{S}$:}\\
\noindent{\bf Input:} $(a_0, a_1, \ldots, a_{n-1}) \in \fq^n$.
\begin{enumerate}
\item If $|S| = 1$ with $S = \{s\}$
\begin{itemize}
\item Define $g: S \to \fq$ by $g(s) = a_0$
\item Return $\EVT{g}{S}$
\end{itemize}
\item Let $\EVT{u_0}{S_0} = \ENTER_{S_0}( a_0, \ldots, a_{\frac{n}{2} -1})$.
\item Let $\EVT{u_1}{S_1} = \EXTEND_{S_0,S_1}(\EVT{u_0}{S_0})$.
\item Let $\EVT{v_0}{S_0} = \ENTER_{S_0}( a_{n/2}, \ldots, a_{n-1})$.
\item Let $\EVT{v_1}{S_1} = \EXTEND_{S_0,S_1}(\EVT{v_0}{S_0})$.
\item Let $$\EVT{\pi}{S} = \EVT{u_0 + X^{n/2} v_0}{S_0} \cup \EVT{u_1 + X^{n/2} v_1}{S_1}.$$
\item Return $\EVT{\pi}{S}$.
\end{enumerate}
% \dan{This is really not important, but a slightly more efficient algorithm would be to multiply $v_0$ by
% $X^{n/2}$ \emph{before} extending, saving a negligible $n/2$ multiplications, but also the precomputation of
% $\EVT{X^{n/2}}{S_1}$.}
% \eli{Dan, then perhaps put it in remark or footnote?}
% \dan{Actually I am completely wrong and this just doesn't work, you can't extend after multiplication
% because the degree is too large. Oops!}

\paragraph{Correctness:}
The correctness follows immediately from the discussion preceding the algorithm.

\paragraph{Running time:}
This algorithm makes two recursive calls to $\ENTER$ on instances of half the size, makes two invocations of $\EXTEND$ on instances of half the size, along with $O(n)$ other operations.
Thus the running time $F(n)$ satisfies:
$$F(n) \leq 2F(n/2) + O(n \log n) + O(n),$$
and thus $F(n) \leq O(n \log^2 n)$, as claimed.
\end{proof}

\subsection{Chinese Remaindering}
\label{sec:chinese remainder}

The following operation receives as input two polynomials $P, Q$ and computes the polynomial $R$ whose remainders modulo two relatively prime polynomials $A,B$ are $P$ and $Q$, respectively.

\begin{theorem}[Chinese Remaindering]
Let $S$ be a basic set with $|S| = n$.
Let $S_0\subseteq S$ be a moiety of $S$.

Let $A(X), B(X)$ be relatively prime polynomials with degrees at most $n/2$.
Suppose that both $A(X)$ and $B(X)$ are half-disjoint from $S$; the moieties
having no zeroes of $A$ and $B$ may be the same moiety for both or a different one for each.

There is an algorithm $\CRT_{S, S_0, A, B}$, which when given as input:
\begin{itemize}
\item $\EVT{P}{S_0}$, where $P(X) \in \fq[X]$ with $\deg(P)<n/2$, and
\item $\EVT{Q}{S_0}$, where $Q(X) \in \fq[X]$ with $\deg(Q)<n/2$,
\end{itemize}
runs in time
$$O(n \log n)$$
and computes $ \EVT{R}{S}$ where $R$ is the unique polynomial of degree $< \deg(A)+\deg(B)$
such that %$P = R \rem A$ and $Q = R \rem B$.
$R \equiv P \pmod A$ and $R \equiv Q \pmod B$.
\end{theorem}
\begin{proof}
By the usual proof of the Chinese Remainder Theorem, the desired
$R(X)$ is of the form:
$$((P(X) \cdot G(X)) \rem A(X))\cdot B(X) + ((Q(X) \cdot H(X)) \rem B(X))\cdot A(X),$$
where $G(X) = (B(X)^{-1})_{A(X)}$, $H(X) = (A(X)^{-1})_{B(X)}$
depend only on $A$ and $B$, and have degrees
$$\deg(G), \deg(H) < n/2.$$

Thus the algorithm simply extends $\EVT{P}{S_0}$ and $\EVT{Q}{S_0}$
to find $\EVT{P}{S}$ and $\EVT{Q}{S}$. Then, using $\EVT{G}{S}$ and $\EVT{H}{S}$
as advice (which can be precomputed, since they only depend on $A$, $B$ and $S$), as well
as $\EVT{A}{S}$ and $\EVT{B}{S}$, we compute:
$$\MODV_{S,A}(\EVT{P\cdot G}{S})\cdot \EVT{B}{S} +  \MODV_{S,B}(\EVT{Q\cdot H}{S})\cdot \EVT{A}{S} $$
% $$\EVT{PG + QH}{S},$$
which is the desired output. Note that $\deg(P\cdot G), \deg(Q\cdot H) < n$, as $\MODV$ requires.

The run-time comes from two invocations of $\EXTEND$, two invocations of $\MODV$, and $O(n)$ other operations,
and is thus $O(n \log n)$ overall.
\end{proof}

\section{Applications to classical problems}
\label{sec:classical}

The previous \cref{sec:algorithms} presented fast algorithms (and arithmetic circuits) for manipulating polynomials represented by their evaluations on basic sets. This section uses those results to
efficiently solve ``classical'' problems of algebraic computation, in which the polynomials are represented in the ``classical'' way, as
sums of monomials. In all cases below we shall transition to
a representation of polynomials by their evaluations
on basic sets, and this will result in running times, \emph{over any polynomially large field}, that are
as good as those of special, classical-FFT-friendly, finite fields.

\subsection{Elementary Symmetric Polynomial Evaluation}
\label{sec:symm poly eval}

\begin{theorem} [Evaluating Elementary Symmetric Polynomials]
Let $t < n < q^{O(1)}$. There is an arithmetic circuit over $\fq$ of size
$$ O( n \log^2 n)$$
which takes as input variables $\alpha_1, \ldots, \alpha_n$ and computes
$$\sym_{n,t}(\alpha_1,\ldots, \alpha_n):= \sum_{J \subseteq [n], |J| = t} \prod_{j \in J} \alpha_j.$$
\end{theorem}
\begin{proof}
We follow the classical approach of computing elementary symmetric polynomials as coefficients of a certain product,
except that we work with polynomials in the new representation.

The idea is to compute the coefficients, in the standard monomial representation, of the polynomial
$$P(X) = \prod_{i = 1}^{n} (X - \alpha_i) =
\sum_{i=0}^n (-1)^{n-i} \cdot \sym_{n,n-i}(\alpha_1,\ldots, \alpha_n) X^{i}$$
We do this by first computing $\EVT{P}{S}$ for a big enough basic set $S$, and then running
$\EXIT_S(P)$ to compute the coefficients of $P$. Details follow.

%\dan{Revise after adding $\MEXTEND$.}
%\swastik{Not sure I like this  ... makes the reader who is interested in
%this algorithm read $\MEXTEND$, which is a strangely defined operation, unlike  $\EXTEND$ which they already know as low degree extension. How about just keeping the algorithm as is, and just adding a remark saying that one can use $\MEXTEND$ to simplify this algorithm and make the running time faster by a factor $2$ (or more).}
%\dan{Accepted; is the footnote below, on $\EXTEND$, OK?}

By adding some dummy $0$ inputs $\alpha_i$ and increasing $n$ by at most a factor $2$, we may assume
that $n$ is a power of $2$.
Next, we claim that $\fq$ can be assumed to contain a basic set of size at least $2n$. Indeed, this
can be done by replacing $\fq$ by an $O(1)$-degree extension of $\fq$ which is sufficiently large,
of size $O(n^2)$, as needed for a basic set of size at least $2n$ to exist in $\fq$
(cf. \cref{sec:EC nice maps}). Moving to a larger $q$  increases the number of arithmetic
operations by a factor of at most $O(1)$ because we assume $n<q^{O(1)}$.

Let $m = \log_{2}(2n)$ and fix arbitrary basic sets $U_0 \subseteq U_1 \ldots \subseteq U_m \subseteq \fq$
with $|U_j| = 2^j$. We shall
compute $\EVT{P}{U_m}$ in a bottom-up manner by computing products of
terms $P_i(X) := X-\alpha_i, i\in [n]$, of increasing size.
We start by computing
$$\EVT{P_1}{U_1},\ldots, \EVT{P_n}{U_1}$$
which takes time $O(1)$ for each term $P_i$ (and total time $O(n)$).

Let $Q(X) = \prod_{i=i_0}^{i_0+2^{j}-1} P_i(X)$ and assume, inductively, that we have
already computed $\EVT{Q'}{U_{j}}$ and $\EVT{Q''}{U_{j}}$ where
$$Q'(X)=\prod_{i=i_0}^{i_0+2^{j-1}-1} P_i(X), \quad Q''(X)=\prod_{i=i_0+2^{j-1}}^{i_0+2^{j}-1} P_i(X).$$  We shall now
compute $\EVT{Q}{U_{j+1}}$ as follows:
\begin{itemize}
	\item Compute $\EVT{Q'}{U_{j+1}}$ using the $\EXTEND$\footnote{Note that all polynomials computed
	in this algorithm are monic and of degrees equal to powers of 2. Thus, as noted in
	\cref{sec:MEXTEND}, it is natural to extend and multiply these polynomials using $\MEXTEND$
	instead of $\EXTEND$, allowing us to take $m = \log_2(n)$, start from evaluations at $U_0$,
	and cut down the running time by a factor of 2.} algorithm
	\item Compute $\EVT{Q''}{U_{j+1}}$ using the $\EXTEND$ algorithm
	\item Pointwise multiply the two to obtain $\EVT{Q}{U_{j+1}}=\EVT{Q'\cdot Q''}{U_{j+1}}$
\end{itemize}

Since $\EXTEND$ runs in time $O(n \log n)$ and pointwise multiplication runs in time $O(n)$, the running time $F(n)$ for this algorithm
satisfies:
$$ F(n) \leq 2 F(n/2) + O(n \log n) \leq O(n \log^2 n).$$

Finally, once we have $\EVT{P}{U_m}$, we can find its standard monomial expansion using $\EXIT_{U_m}$, which also runs in  time $O(n \log^2 n)$.

The desired output $\sym_{n,t}$ is one of the coefficients in this standard monomial expansion, and is thus computed in time $O(n \log^2 n)$, as claimed.
\end{proof}

\subsection{Multipoint evaluation over general sets of points}
\label{sec:multipoint eval general sets}

Previously we evaluated polynomials over basic sets in quasi-linear time (see \cref{thm:enter}).
The next result shows that evaluating polynomials over general sets of points can also be done in (slightly worse) quasi-linear time.

\begin{theorem} [Multipoint polynomial evaluation]
	\label{thm:multipoint evaluation}
Assume $n<q$. Given any set $B$ of $m$ points in $\fq$, there exists an arithmetic circuit over $\fq$ (that depends on $B$) of size
$$ O( n \log^2 n + m \log^2 m )$$
which takes as input $(a_0, \ldots, a_{n-1}) \in \fq^n$ and computes
$\EVT{P}{B}$
for $P(X) = \sum_{i=0}^{n-1} a_i X^i$.
\end{theorem}

\begin{remark}\label{rem:mulpoleval}
Note that if $n\geq q$ then we can first reduce $P(X)$ modulo $X^q - X$ trivially in $n$ steps and get
back to the case $n < q$.
Moreover, if $m>n$, then we can partition $B$ into $O(m/n)$ sets of size at most $O(n)$,
getting a run-time of $O(m \log^2 n)$,
whereas if $m<n$, then we can decompose $P$ into $O(n/m)$ polynomials of degree at most $O(m)$,
getting a run-time of $O(n \log^2 m)$.
\end{remark}

\begin{proof}
	Let $B=\{b_1,\ldots, b_m\}$.
The idea of the algorithm is based on the fact that
$$P(b_i) = P(X) \rem (X-b_i).$$
To find $P(X) \rem (X- b_i)$, we start with $P(X) \rem \prod_{i=1}^m(X- b_i)$, and successively
compute $P(X) \rem \prod_{i \in I} (X-b_i)$ for smaller and smaller sets $I \subseteq [m]$. Details follow.

The algorithm starts by running $\ENTER_{U_a}$ to find
$$\EVT{P}{U_a},$$
where $a = \log_2 n + O(1)$. This step runs in $O(n \log^2 n)$ time.
% Note that for any
% $b \in B \cap U_a$, the evaluation table $\EVT{P}{U_a}$ contains the value $P(b)$ directly,
% so we are left with computing the values of $P$ on the smaller set $B \setminus U_a$. Thus we may assume
% henceforth that $B$ is disjoint from $U_a$.
% \dan{Actually this isn't good, since it makes the CRT of 7.3 inapplicable. bah}

Next, we tweak $B$ and $U_a$ until they are of similar sizes, specifically, $2^{a-2} < |B| \le 2^{a-1}$.

In the case $|B| \le 2^{a-2}$, let $a' =\ceil{\log_2 m} +1 < a$.
We wish to assume that $B$ is half-disjoint from $U_a$: if it is not the case,
we may simply split $B$ into two parts that are each half-disjoint, e.g.\@ $B\cap U_{a-1}$ and
$B \setminus U_{a-1}$.
% This splitting may double our total work (once), which would not change
% the big-$O$ complexity of the algorithm.
Then, assuming half-disjointness, we may run
$\MODV_{U_a, \prod_{b \in B} (X- b)}(\EVT{P}{U_a})$ in $O(n \log n)$ time to obtain
$$\EVT{P \rem \prod_{b \in B} (X- b)}{U_{a}}.$$
% This is possible since $\prod_{b \in B} (X- b)$ is disjoint from $U_a$ and has degree at most
% $\frac{|U_a|}{2}$ (in fact, much smaller) and takes $O(n \log n)$ time.
The resulting polynomial will have degree strictly less than
$|B| \le 2^{a'-1}$, and we may restrict its evaluation table
$\EVT{P \rem \prod_{b \in B} (X- b)}{U_{a}}$ to
$$\EVT{P \rem \prod_{b \in B} (X- b)}{U_{a'-1}}$$
at no cost while maintaining the fact that it represents $P \rem \prod_{b \in B} (X- b)$.
We then continue to evaluate this polynomial on $B$, replacing $a$ with $a'$, and noting that
$2^{a'-2} < |B| \le 2^{a'-1}$, and $a' \le \min(\log_2(n), \log_2(m)) + O(1)$.

In the case $|B| > 2^{a-1}$, split $B$ arbitrarily into $l$ disjoint parts, each of size at
most $2^{a-1} = O(n)$, and proceed on each part separately. As in the previous case we further require
that each part be half-disjoint from $U_a$, and observe that again this requires at most one additional
part (e.g.\@ by taking one of the parts equal to $B \cap U_{a-1}$), and can be achieved using only
$l = O(\frac{m}{n} + 1)$ parts, and note that this bound also covers the previous case (with 1 or 2
parts).
The complexity of the remaining work done on each part will be multiplied by $l$ to obtain the total
complexity. We continue now with
$|B|$ denoting a single part, of size at most $2^{a-1}$, and half-disjoint from $U_a$.
Again we have $a \le \min(\log_2(n), \log_2(m)) + O(1)$. As in the previous case, the next step
is to run $\MODV_{U_a, \prod_{b \in B} (X- b)}(\EVT{P}{U_a})$ and restrict to $U_{a-1}$, yielding
$$\EVT{P \rem \prod_{b \in B} (X- b)}{U_{a-1}}$$
in $O(n \log n)$ time.

We can now get the desired result by applying the following recursive step, for $j=a-1,a-2,\dots,1$:
Suppose $A(X)$ is a product of $\le 2^j$ different linear factors. Then we may write
$A(X) = A'(X) \cdot A''(X)$, where $\deg(A'), \deg(A'') \le 2^{j-1}$, and both $A', A''$ are
half-disjoint from $U_j$. Then given $\EVT{P \rem A}{U_j}$,
we can compute
$$\EVT{P \rem A'}{U_{j}} = \MODV_{U_j, A'}(\EVT{P\rem A}{U_j})$$
$$\EVT{P \rem A''}{U_{j}} = \MODV_{U_j, A''}(\EVT{P\rem A}{U_j})$$
in time $O(|U_j| \log |U_j|)$, and then restrict the tables to $\EVT{P \rem A'}{U_{j-1}}$,
$\EVT{P \rem A''}{U_{j-1}}$.

At layer $j$ of the recursion we perform $2^{a-j}$ $\MODV_{U_j}$ operations, taking a total run time
of $O(|U_a|\log|U_j|)$, and summing over all layers $j$ we get a run time of
$$O(|U_a| \log^2 |U_a|) = O(\min(n \log^2 n,m \log^2 m))$$
per part. Multiplying by the number of parts $l = O(\frac{m}{n} + 1)$ and adding
the $O(n \log^2 n)$ from $\ENTER$, we get that the total run time is
$$ O(n \log^2 n + m \log^2 m),$$
as claimed.
\end{proof}

\subsection{Interpolation from general evaluation sets}\label{sec:interpolation gen sets}

In \cref{sec:exit} we showed how to interpolate in quasi-linear time from evaluations on basic sets.
The following result, the converse of the previous \cref{thm:multipoint evaluation},
obtains quasi-linear running time (with somewhat worse parameters) for interpolating from general evaluation sets.

\begin{theorem}[Polynomial interpolation from general evaluation sets]
	\label{thm:interpolation gen sets}
Let $B \subseteq \fq$ be a set of $m$ points.
There is an arithmetic circuit over $\fq$ (depending on $B$) of size
$$ O( m \log^2 m)$$
which takes as input an evaluation table $\EVT{\pi}{B}$
% \david{$p$ seems more natural to me here than $C$, WDYT?}
% \dan{I think most algorithm in section 6 use $\EVT{\pi}{S}$, so I vote $\pi$ instead of $p$ or $C$}
and computes
the coefficients $a_i$ of the unique polynomial of degree $<m$:
$$P(X) = \sum_{i=0}^{m-1} a_i X^i,$$
such that $\EVT{P}{B}=\EVT{\pi}{B}$.
\end{theorem}

\begin{proof}[Proof Sketch]
Since this algorithm is roughly the opposite of the previous algorithm,
instead of applying $\MODV$ in each recursive step as done above,
we use $\CRT$ to
do fast Chinese remaindering to compute $\EVT{P}{U}$ for a basic set $U$, followed by
calling  $\EXIT(\EVT{P}{U})$ to get the desired standard polynomial representation. The running time and analysis are similar to that of the previous \cref{thm:multipoint evaluation}.
\end{proof}

%\swastik{Maybe a low degree extension theorem too?}

\section*{Acknowledgements}
We thank an anonymous reviewer for informing us about the previous work by Chudnovsky and
Chudnovsky~\cite{ChCh}, %of which we were not previously aware,
as well as other valuable references and comments.
Some of this research was done while SK was visiting StarkWare in 2019. SK is grateful to StarkWare for the warm hospitality and the electrifying atmosphere.

\newpage
\bibliographystyle{alpha}
\bibliography{ecFFT}{}

\newpage

\appendix

\section{The paper of Chudnovsky and Chudnovsky~\cite{ChCh}}
\label{sec:ChCh-comp}
Some of the core ideas appearing in this paper were first suggested, in brief,
in the final section of a paper by Chudnovsky and Chudnovsky \cite[Section 6]{ChCh}.

%We were made aware of this paper by an anonymous reviewer,
%who further remarked that the authors only sketch
%their idea. Indeed, the writing is extremely succinct, many details are omitted or only
%hinted at, and we believe there is at least one substantial error. 

The main claim from~\cite{ChCh} relevant for us is an ``Elliptic Interpolation Algorithm'', which
describes how to use elliptic curve groups over
finite fields to solve a certain rational function interpolation problem via an FFT-type algorithm.
They call this rational function interpolation problem and the algorithm for it
the Fast Elliptic Number Theoretic Transform (FENTT).
%Specifically, for all primes $p$ and such that the interval $[p \pm 2\sqrt{p} + 1]$ contains a power of
%2, 
Specifically for any prime $p$
and for a certain $n$ and certain fixed $\alpha_1, \ldots, \alpha_n, \beta_1, \ldots, \beta_n \in \F_p$
the following holds:
\begin{itemize}
\item Given as input $x_1, \ldots, x_n$, consider the rational function
$$F(Z) = \sum_{i} \frac{x_i}{Z - \alpha_i}.$$
There is an algorithm to compute $\langle F(\beta_j) \rangle_{j \in [n]}$ in $O(n \log n)$ operations.
\end{itemize}
Succinctly, the claim is that there exists an $n \times n$ Cauchy matrix $M$ with entries in $\F_p$
such that multiplying by $M$ can be done in $O(n \log n)$ operations. The authors go on to suggest
that $n$ can be taken to be $\Omega(p)$ whenever there is a power of $2$ within $p+1 \pm 2\sqrt{p}$, but we
believe there is an error in the analysis of their algorithm, and that in fact, there are no elliptic curves
that can make their algorithm achieve such large $n$ uniformly for all $p$ --- like the FFT, the
$n$ in the FENTT is limited by the factorization of $p-1$.

Technically, the proposal is based on the doubling map, which is a specific $4$-isogeny mapping an
elliptic curve $E$ to itself (where we use general 2-isogenies between different curves).
The overall scheme then follows the high level outline of \cref{sec:ode to EC},
albeit with a more direct translation, leading to a solution of the rational function interpolation
problem rather than the polynomial extension problems that we aim for (with an eye on applications).

A key point that~\cite{ChCh} overlooked is that the kernel of the doubling map is
isomorphic to $\Z/2\Z \times \Z/2\Z$, not $\Z/4\Z$. Thus the doubling map (which is a degree $4$ map) will be
a $4$-to-$1$ map on a subgroup $G  \subseteq E(\F_p)$ if and only if $\Z/2\Z \times \Z/2\Z$ is a subgroup of $G$.
Thus if we seek a sequence $G_0, G_1, \ldots, G_i, \ldots, G_t$ of subgroups of $E(\F_p)$
such that each $G_{i+1}$ is the image of $G_i$ under doubling, with $|G_{i+1}| = |G_i|/4$,
then we need $\Z/2\Z \times \Z/2\Z$ to be a subgroup of each $G_i$.
Since  $E(\F_p)$ is a rank 2 abelian group, the above sequence of groups can exist only if
$\Z/2^{t}\Z \times \Z/2^{t}\Z$ is a subgroup of $E(\F_p)$.
%(not just {\em any} subgroup of cardinality $4^t$, which is what can be guaranteed, via theorems of Deuring/Waterhouse, by considering the number of points in $E(\F_p)$).
This is a much more stringent requirement on $E(\F_p)$ than merely having size divisible by $4^t$ (which is what one can get out of theorems of Deuring/Waterhouse).
In fact, a theorem about the possible structures of an elliptic curve group, proved independetly by
Ruck~\cite{Ruck} and Voloch~\cite{Voloch}, implies that an elliptic curve over $\F_p$ can contain a
subgroup isomorphic to $\Z/2^{t}\Z \times \Z/2^{t}\Z$ only if $p \equiv 1 \pmod{2^{t}}$.
 {\em Thus, the size of the largest FENTT supported by $\F_p$ is bounded in terms of the factorization of $p-1$},
and is at most quadratically larger than the size of the largest FFT supported by $\F_p$.

This is quite restrictive; for example, if $p \equiv 3 \pmod 4$, then the largest FENTT 
supported by $\F_p$ has instance size $4$. In contrast, our approach works with
$\Omega(\sqrt{p})$ instance size over $\F_p$ for all primes $p$.

%then  such that iterated application of the doubling map keeps shrinking the size of the
%group by a factor of $4$, we require $E(\F_p)$ to have the product group
%$\Z/2^{t/2}\Z \times \Z/2^{t/2}\Z$ as a subgroup  (not just {\em any} subgroup of cardinality $2^t$,
%which is what our argument requires).
%This is a more stringent requirement on $E(\F_p)$, and merely having size divisible by $2^t$ does not
%guarantee the existence of such a large product subgroup.
%On the contrary, theorems of Ruck and Voloch imply that an elliptic curve over $\F_p$ can contain a subgroup isomorphic to
%$\Z/2^{t/2}\Z \times \Z/2^{t/2}\Z$ only if $p \equiv 1 \pmod{2^{t/2}}$.
%{\em Thus, the size of the largest FENTT supported by $\F_p$ is bounded in terms of the factorization of $p-1$},
%and is at most quadratically larger than the size of the largest FFT supported by $\F_p$.

%The existence of such $E$ where $E(\F_p)$ has a large product subgroup is not covered
%by the theorems of Deuring/Waterhouse (\cref{sec:group struct,thm:Waterhouse}).
%, and requires further argument.
%\dan{added the next clause; is it too harsh?}

The necessity and feasibility of this product subgroup structure condition are not mentioned in~\cite{ChCh};
in fact, the paper erroneously states that every elliptic curve with order divisible by $2^t$ will
have points of order $2^t$, i.e.\@ a {\em cyclic} subgroup of order $2^t$,
which is not necessarily the case, and further, because of the confusion regarding $\Z/2\Z \times \Z/2\Z$ and $\Z/4\Z$
described above, is not relevant to success of the FENTT.

%In fact, for general primes it is often not possible to find arbitrarily large product subgroup.
%A theorem of {} (see \cite{}) implies that a curve over $\F_p$ can contain a subgroup isomorphic to
%$\Z/2^{t/2}\Z \times \Z/2^{t/2}\Z$ only if $p \equiv 1 \pmod{2^{t/2}}$. In other words, a FENTT
%of size $2^t$ is only possible for primes where the regular multiplicative FFT is possible
%on at least $2^{t/2}$ variables.

In summary, the idea of using elliptic curve groups in place of multiplicative groups for computing
FFT-like transforms to all prime fields is not new to our paper, but we believe that an implementation 
that supports transforms with large instance sizes over {\em all} primes $p$ is.
Furthermore, our approach is intentionally adapted to a
certain flexible polynomial extension problem (rather than the specific rational function evaluation
considered in~\cite{ChCh}), and this is what enables our applications to data structures and
algorithms for intrinsically interesting classical problems involving polynomials.

% We thank the anonymous reviewer for making us aware of this paper and other valuable references and comments.

\section{Proof of the decomposition lemma \ref{lem:decomposition}}
\label{sec:proof-of-decomposition}

\decomp*

\begin{proof}
For general $P_i(Y) \in V_{d/\delta}$, where
$$P_i(Y) = \sum_{j=0}^{d/\delta-1} a_{ij} Y^j,$$
consider the polynomial
$$P(X) = \sum_{i=0}^{\delta-1} X^i  P_i(\psi(X)) \cdot v(X)^{\frac{d}{\delta}-1}.$$

Observe that
\begin{align}
	\nonumber
P(X) &= \sum_{i=0}^{\delta-1} X^i P_i(u(X)/v(X)) \cdot v(X)^{\frac{d}{\delta}-1}
=\sum_{i=0}^{\delta-1} X^i \sum_{j=0}^{d/\delta-1} a_{ij}(u(X)/v(X))^j \cdot v(X)^{\frac{d}{\delta}-1}\\
&=\sum_{i=0}^{\delta-1} \sum_{j=0}^{d/\delta-1} a_{ij} X^i u(X)^j \cdot v(X)^{\frac{d}{\delta}-1 - j}, \label{eq:decomp}
\end{align}
and thus $P(X) \in V_d$.
% since the general summand has degree
% $i +j\deg(u) +(\tfrac{d}{t}-1-j) \le t-1 + (\tfrac{d}{t}-1)t = d-1 $.
We shall use the following claim, proved below:
\begin{claim}\label{clm:proper decomposition}
For every choice of $P_0(X), P_1(X), \ldots, P_{\delta-1}(X) \in V_{d/\delta}$, not all $P_i$ being zero,
the polynomial:
$$P(X) = \sum_{i=0}^{\delta-1} X^i  P_i(\psi(X)) \cdot v(X)^{\frac{d}{\delta}-1}$$
is nonzero.
\end{claim}
Together with the fact that the dimension of $\left(V_{d/\delta}\right)^\delta$ equals the dimension of $V_d$, the theorem follows.
\end{proof}

\begin{proof}[Proof of \cref{clm:proper decomposition}]
Reordering the right hand side of \cref{eq:decomp}
gives
	\begin{align*}
		P(X)
		= \sum_{j=0}^{d/\delta-1}    Q_j(X) u(X)^j v(X)^{d/\delta-1-j},
	\end{align*}
	where $Q_j(X) = \sum_{i=0}^{\delta-1} a_{ij} X^i$ is a polynomial of degree $< \delta$.

	Since $\deg(\psi) = \delta$, we have that either $\deg(u(X))= \delta$ or $\deg(v(X)) = \delta$.
	Suppose $\deg(u(X)) = \delta$, the other case being similar.

	The assumption that not all $P_i(X)$ are zero implies that not all $Q_j(X)$ are zero, so
	let $j_0$ be the minimal integer such that  $Q_{j_0}(X)$ is a nonzero polynomial.
	Then $P(X)$ is divisible by $u(X)^{j_0}$, and
	$$ \frac{P(X)}{u(X)^{j_0}} =
	\sum_{j=j_0}^{d/\delta-1}    Q_j(X) u(X)^{j-j_0} v(X)^{d/\delta-1-j},$$

	Finally, we observe this polynomial is {\em nonzero} modulo $u(X)$,
	since modulo $u(X)$ it equals:
	$$ Q_{j_0}(X) \cdot v(X)^{d/\delta-1-j_0},$$
	$v(X)$ is invertible modulo $u(X)$ (since $v(X)$ is relatively prime to $u(X)$), and $Q_{j_0}(X)$ is nonzero modulo $u(X)$ because it is a nonzero  polynomial of degree strictly less than $\delta$.
	This implies that $P(X)$ is a nonzero polynomial, completing the proof of
	the claim.
\end{proof}

\section{Proofs from \cref{sec:FFTrees from EC}}\label{sec:elliptic appendix}

\subsection{Proof of \cref{prop:standard form}}\label{sec:psi appendix}

\standard*

\begin{proof}
	As noted in \cref{sec:group law}, $\pi'(Q)=\pi'(-Q)$ for all points $Q\in E'$.
	In fact, this equality holds for all the points of $E'(\fqbar)$ -- the set of all solutions
	of the curve equation in the algebraic closure of $\fq$ (i.e., considering all solutions over all the finite field extensions of $\fq$).
	The composition $\pi'\circ\phi:E \to \Prj$ can be represented as an element of $\fq(X)[Y]/F(X,Y)$
	where $F(X,Y)=0$ is the equation that defines $E$ (see \cref{eq:weierstrass}).
	Notice that $\fq(X)[Y]/F(X,Y)$ is a degree $2$ extension field of $\fq(X)$
	because $F$ is a degree $2$ polynomial in $Y$ with coefficients in $\fq(X)$, so we can write $\pi'\circ\phi(X,Y) = \psi(X) + Y \cdot \chi(X)$ for some $\psi, \chi\in \fq(X)$.
	We know that $\phi$ is a group homomorphism, so $\pi'(\phi(-Q))=\pi'(-\phi(Q))=\pi'(\phi(Q))$ for all points
	$Q\in E(\fqbar)$. In particular, since $Q$ and $-Q$ have the same $x$ coordinate but different $y$
	coordinates (unless $Q=-Q$), then $\chi(x)=0$ for every $x$ coordinate of a point in $E(\fqbar)$
	except for at most $4$ points (see \cite[Exercise 3.7]{Silverman} or \cite[Example 2.5]{Wash08}). %\eli{why 4?}
	Since there are infinitely many such points over the algebraic closure $\fqbar$
	we conclude that $\chi$ is the constant
	$0$ function and $\pi'\circ\phi(x,y) = \psi(x)$ or equivalently $\pi'\circ\phi = \psi\circ\pi$.
\end{proof}

\subsection{Existence of an appropriate $G_0$ in the proof of \cref{thm:curve sequence}}
\label{sec:existence of G_0}
Recall that we have constructed a curve $E_0$ of size $N$ with $\nn \mid N$ and $N > 2\nn$, where
$\nn$ is a power of 2. Our goal in this section is to show that there exists a subgroup $G_0 < E_0$
which is of size $\nn$, and such that there exists a coset $C$ of $G_0$ with $C \neq -C$.

As noted in \cref{sec:group struct}, $E_0$ is of
rank at most 2, and there is an isomorphism
$$\tau : E_0  \leftrightarrow \Z/(m_12^{l_1}\Z) \times \Z/(m_22^{l_2}\Z) $$
where $m_1, m_2$ are odd with $m_1 \mid m_2$,
$l_1 \le l_2$, $m_1 m_2 2^{l_1+l_2} = N$ and in particular $l_1 + l_2 \ge \rounds$. A subgroup $G_0$
of size $\nn$ will necessarily be of the form
$$G_0 =
\tau^{-1}\parens*{(m_12^{l_1-k_1}\Z)/(m_12^{l_1}\Z) \times (m_2 2^{l_2- k_2}\Z)/(m_22^{l_2}\Z)}
\simeq \Z/2^{k_1}\Z \times \Z/2^{k_2}\Z $$
with $k_1\le l_1$, $k_2\le l_2$ and $k_1 + k_2 = \rounds$, and the quotient $E/G_0$ is then
isomorphic to
$$E_0/G_0 \simeq \Z/(m_12^{l_1-k_1}\Z) \times \Z/(m_22^{l_2-k_2}\Z).$$
We wish to ensure that this group contains an element $C$ such that $C \neq -C$, or equivalently,
$2C \neq 0$. This is clearly the case for any choice of $k_1, k_2$, except if
$m_1 = m_2 = 1$ and $l_1 - k_1, l_2 - k_2 \le 1$, which are the cases where $E_0/G_0$
is isomorphic to either the trivial group, $\Z/2\Z$, or $\Z/2\Z \times \Z/2\Z$. But since
$m_1 m_2 2^{l_1 -k_1 + l_2 - k_2} = \frac{N}{\nn} > 2$,
this happens only when $N = 4\nn$ and for the choice $k_1 = l_1 -1$ and $k_2 = l_2 - 1$.
But, by the assumption $\nn > 1$ and by $l_2 \ge l_1$, we find $l_2 \ge 2$, thus we may choose
instead $k_1 = l_1$ and $k_2 = l_2 - 2$, to obtain $E_0/G_0 \simeq \Z/4\Z$, which indeed contains
an element $C$ with $C \neq -C$. \qed

\end{document}